\documentclass[10pt]{article}
\usepackage{fullpage}
\usepackage[bookmarks]{hyperref}
\usepackage{amssymb}
\usepackage{amsmath}
\usepackage{amsthm}
\usepackage{graphicx,color,colordvi}
\usepackage{bbm}
\usepackage{cleveref}
\usepackage{stmaryrd}
\usepackage[utf8]{inputenc}
\usepackage{fullpage}
\usepackage[blocks]{authblk}
\usepackage{dsfont}
\usepackage{mathtools}
\usepackage{authblk}
\newcommand{\norm}[1]{\left\| #1 \right\|}
\usepackage{ulem}

\usepackage{pgfplots}

\def\openone{\leavevmode\hbox{\small1\kern-3.8pt\normalsize1}}

\def\cn{{\cal N}}

\def\CC{\mathbb{C}}
\def\RR{\mathbb{R}}
\def\ZZ{\mathbb{Z}}

\def\NN{\mathbb{N}}
\def\LL{\mathbb{L}}

\def\11{\mathbb{I}}
\def\LL{\mathcal{L}}

\def\ind{1\hspace{-0.27em}\mathrm{l}}

\newtheorem{theorem}{Theorem}
\newtheorem{lemma}{Lemma}
\newtheorem{proposition}{Proposition}
\newtheorem{corollary}{Corollary}

\theoremstyle{definition}
\newtheorem{definition}{Definition}
\newtheorem{example}{Example}

\newcommand{\proj}[1]{\ket{#1}\!\bra{#1}}

\def\eps{\varepsilon}

\newcommand{\qty}[1]{\left\lbrace #1   \right\rbrace}
\newcommand{\supp}{\mathop{\rm supp}\nolimits}

\newcommand{\tr}{\mathop{\rm Tr}\nolimits}

\newcommand{\ds}{\displaystyle}

\usepackage{pst-node}
\usepackage{tikz-cd} 

\newcommand{\bra}[1]{\langle#1|}
\newcommand{\ket}[1]{|#1\rangle}

\newcommand{\cA}{{\cal A}}
\newcommand{\cB}{{\cal B}}
\newcommand{\cD}{{\cal D}}

\newcommand{\cF}{{\mathcal{F}}}

\newcommand{\cN}{{\cal N}}

\newcommand{\cH}{{\cal H}}
\newcommand{\cK}{{\cal K}}

\newcommand{\cP}{\mathcal{P}}

\newcommand{\cL}{{\cal L}}

\newcommand{\cR}{{\cal R}}
\newcommand{\cX}{{\cal X}}

\newcommand{\cM}{{\mathcal{M}}}

\newcommand{\Id}{{\mathds{1}}}

\newcommand{\C}{{\mathbb{C}}}

\def\e{\mathrm{e}}

\usepackage{graphicx}
\usepackage{setspace}
\usepackage{verbatim}
\usepackage{subfig}

\theoremstyle{definition}

\theoremstyle{remark}
\newtheorem{remark}{Remark}

\numberwithin{equation}{section}

\DeclareRobustCommand\openone{\leavevmode\hbox{\small1\normalsize\kern-.33em1}}

\newcommand{\id}{{\rm{id}}}
\newcommand{\be}{\begin{equation}}
	\newcommand{\ee}{\end{equation}}
\newcommand{\bea}{\begin{eqnarray}}
	\newcommand{\eea}{\end{eqnarray}}
\newcommand{\beas}{\begin{eqnarray*}}
	\newcommand{\eeas}{\end{eqnarray*}}

\title{Approximate tensorization of the relative entropy \\
	for noncommuting conditional expectations}

\begin{document}

 	\bibliographystyle{abbrv}

\author[1]{Ivan Bardet}
\affil[1]{\small{Inria Paris, France}}
\author[2,3,4,5]{{\'A}ngela Capel}
\affil[2]{Departamento de An{\'a}lisis Matem{\'a}tico y Matem{\'a}tica Aplicada, Universidad Complutense de Madrid, Spain} 
\affil[3]{Instituto de Ciencias Matemáticas (CSIC-UAM-UC3M-UCM),  Madrid, Spain }
\author[4,5]{Cambyse Rouz\'{e}}
\affil[4]{Department of Mathematics, Technische Universit\"at M\"unchen, 85748 Garching, Germany}
\affil[5]{Munich Center for Quantum Science and Technology (MCQST), M\"unchen, Germany}

\date{}

\maketitle

\begin{abstract}
In this paper, we derive a new generalisation of the strong subadditivity of the entropy to the setting of general conditional expectations onto arbitrary finite-dimensional von Neumann algebras.  This generalisation, 
referred to as approximate tensorization of the relative entropy, consists in a lower bound for
the sum of relative entropies between a given density and its respective projections onto two intersecting
von Neumann algebras in terms of the relative entropy between the same density and its projection
onto an algebra in the intersection, up to multiplicative and additive constants. In particular, our inequality reduces to
the so-called quasi-factorization of the entropy for commuting algebras, which is a key step in modern proofs of the logarithmic Sobolev inequality for classical lattice spin systems. We also provide estimates on the constants in terms of conditions of clustering of correlations in the setting of quantum lattice spin systems. Along the way, we show the equivalence between conditional expectations arising from Petz recovery maps and those of general Davies semigroups.
	\end{abstract}
\tableofcontents
\section{Introduction}
In the last few decades, entropy has been proven to be a fundamental object in various fields of mathematics and theoretical physics. Its quantum analogue characterizes the optimal rate at which two different states of a system can be discriminated when an arbitrary number of copies of the system is available. Given two states $\rho,\sigma$ of a finite-dimensional von Neumann algebra $\cN\subset \cB(\cH)$, it is given by 
\begin{align*}
D(\rho\|\sigma):=
\tr[\rho\,(\ln\rho-\ln\sigma)]\,,
\end{align*} 
whenever $\supp(\rho)\subset\supp(\sigma)$, where $\tr$ denotes the unnormalized trace on $\cB(\cH)$. When $\sigma:=\Id_{\cH}/d_\cH$ is the completely mixed state of $\cB(\cH)$, the relative entropy can be written in terms of the von Neumann entropy $S(\rho):=-\tr[\rho\ln\rho]$ of the state $\rho$:
\begin{align*}
D(\rho\|\Id_\cH/d_\cH)=-S(\rho)+\ln(d_\cH)\,.
\end{align*}

Probably the most fundamental property of entropy is the following \textit{strong subadditivity} inequality (SSA) \cite{lieb1973SSA}: given a tripartite system $\cH_{ABC}:=\cH_A\otimes \cH_B\otimes \cH_C$ and a state $\rho\equiv \rho_{ABC}$ on $\cH_{ABC}$, 
\begin{align}\tag{SSA}
S(\rho_{ABC})+S(\rho_B)\le S(\rho_{AB})+S(\rho_{BC})\,,
\end{align}
where for any subsystem $D$ of $ABC$, $\rho_{D}:=\tr_{D^c}[\rho_{ABC}]$ denotes the marginal state on $D$. Restated in terms of the quantum relative entropy, (SSA) takes the following form:
\begin{align}\label{SSA}
D\left(\rho_{ABC}\Big\|\rho_B\otimes \frac{\Id_{AC}}{d_{\cH_{AC}}}\right)\le D\left(\rho_{ABC}\Big\| \rho_{AB}\otimes \frac{\Id_{C}}{d_{\cH_C}}  \right)+D\left(\rho_{ABC}\Big\| \rho_{BC}\otimes \frac{\Id_{A}}{d_{\cH_{A}}}  \right)\,.
\end{align}

In the present paper, we consider the following more general framework: let $\cM\subset \cN_1,\cN_2\subset\cN$ be four von Neumann subalgebras of the algebra of linear operators acting on a finite-dimensional Hilbert space $\cH$, and let $E^\cM,E_1,E_2$ be conditional expectations onto $\cM,\cN_1,\cN_2$, respectively. When the quadruple $(\cM,\cN_1,\cN_2,\cN)$ forms a \textit{commuting square}, that is when $E_{1}\circ E_2=E_2\circ E_1=E^\cM$, the following generalization of SSA occurs: for any state $\rho$ on $\cN$,
\begin{align}\label{eqcoarsegrain}
D(\rho\|E^{\cM}_*(\rho))\le D(\rho\| E_{1*}(\rho))+D(\rho\|E_{2*}(\rho))\,,
\end{align}
where the maps $E^\cM_*,E_{1*}$, $E_{2*}$ are the Hilbert-Schmidt duals of $E^\cM, E_{1}, E_2$, also known as \textit{coarse-graining maps} \cite{petz2008}. One can easily recover the previous (SSA) inequality from (\ref{eqcoarsegrain}) by taking $\cN\equiv\cB(\cH_{ABC})$, and the coarse-graining maps to be the partial traces onto the subalgebras $\cN_1\equiv\cB(\cH_{AB}) $, $\cN_2\equiv \cB(\cH_{BC})$ and $\cM\equiv \cB({\cH_B})$, respectively. Thus, inequality \eqref{eqcoarsegrain} can be seen as an operator algebraic generalization of the (SSA) inequality.

However, the commuting square assumption and subsequently inequality \eqref{eqcoarsegrain} are not satisfied in most of the cases of interest that appear in information-theoretical settings or quantum many-body systems. Indeed, in the context of interacting lattice spin systems, conditional expectations arising e.g. from the large time limit of a dissipative evolution on subregions of the lattice generally do not satisfy the commuting square assumption. In this case, approximations of the (SSA) were found in the classical case (i.e. when all algebras are commutative) and when $\cM\equiv \CC\Id_{\cH}$ \cite{cesi2001quasi}. For classical lattice spin systems, these inequalities, termed as \textit{approximate tensorization of the relative entropy} (also known in the literature as \textit{quasi-factorization of the relative entropy} \cite{cesi2001quasi,[D02]}), take the following form
\begin{align}\label{eq:Cesi-quasi-factorization}
D(\rho\|\sigma)\le \frac{1}{1-2c_1}\,\big(D(\rho\|E_{1*}(\rho))+D(\rho\|E_{2*}(\rho))\big)\,,
\end{align}
where $\sigma:= E^{\cM}_*(\rho)$ for all states $\rho$, and $c_1:=\|E_{1}\circ E_{2}-E^\cM:\,\mathbb{L}_1(\sigma)\to \mathbb{L}_\infty(\cN)\|$ is a constant that measures the violation of the commuting square condition for the quadruple $(\cM,\cN_1,\cN_2,\cN)$. For reasons that will become clear in the remaining parts of the article, we refer to the constant $c_1$ as the \textit{clustering of correlations} constant in this introduction.


An inequality of the form of \eqref{eq:Cesi-quasi-factorization} is the main ingredient in modern proofs of \textit{modified logarithmic Sobolev inequalities} (MLSI) which govern the rapid thermalization of classical lattice spin systems evolving according to a Glauber dynamics and in the high temperature regime \cite{cesi2001quasi},  \cite{[D02]}. Furthermore, the aforementioned quantum versions of \eqref{eq:Cesi-quasi-factorization} for different \textit{conditional relative entropies} have been used in the past years to obtain some examples of positive MLSI for quantum spin systems \cite{[CLP18a], BardetCapelLuciaPerezGarciaRouze-HeatBath1DMLSI-2019, [CLP18]}. Our main motivation in the current paper is a continuation of those results by further generalizing  \eqref{eq:Cesi-quasi-factorization} to a more abstract setting, with the aim of providing new interesting examples of positive MLSI. In fact, after the first version of this manuscript, the main results contained here have allowed some of the authors to solve a long-standing open problem regarding a system-size independent MLSI for certain evolutions that converge to Gibbs states of nearest-neighbour commuting Hamiltonians at high enough temperature in \cite{capel2020modified}. 

\paragraph{Main results:}In this paper, building on the previous results of approximate tensorization of the form of \eqref{eq:Cesi-quasi-factorization},  we take one step further and introduce a \textit{weak approximate tensorization} for the relative entropy, denoted throughout the text by AT$(c,d)$, which amounts to the existence of positive constants $c\geq 1$ and $d\ge 0$ such that (see \Cref{theo_AT_pinching})\footnote{The definition of (strong) approximate tensorization recently arose in a first version of the paper \cite{nick2019scoopingpaper}, where it was coined as ``adjusted subadditivity of relative entropy''. As explained by the author himself, this definition was already present in an earlier draft of our present article, which we had shared with him (see also the recently published thesis \cite{thesisangela}). Furthermore, the techniques that we introduce here are different from his, and more in line with the classical literature on the subject.}
	\begin{align}\tag{$\operatorname{AT(c,d)}$}
D(\rho\|E^\cM_*(\rho))\le c\,\big(D(\rho\|E_{1*}(\rho))+D(\rho\|E_{2*}(\rho))\big)+d\,.
\end{align}
Whenever $d=0$, we refer to the previous bound as a \textit{strong approximate tensorization} for the relative entropy. Nevertheless, as opposed to the classical setting, conditional expectations arising from dissipative evolutions on quantum lattice spin systems generically do not satisfy the commuting square condition even at infinite temperature. This difference is exclusively due to the non-commutativity of the underlying algebras. The additive constant $d$ is meant to take into account this correction from the classical case.

Note that, at infinite temperature, the conditional expectations are selfadjoint with respect to the Hilbert-Schmidt inner product, a property referred to as \textit{symmetric} in \cite{gao2018fisher,bardet2021group}. Under this condition, in \cite{gao2017strong}, a different extension of (SSA) was proposed. In our framework, the inequality derived in \cite{gao2017strong} leads to an AT$(1,d)$, which can be regarded as measuring the violation of the commutative square condition at infinite temperature. On the other hand, our strong approximate tensorization constant $c$ can be regarded as a finite temperature relaxation of the case $c=1$ in \cite{gao2017strong}. 


The first AT$(c,d)$ inequality that we obtain is presented in \Cref{HSAT}, where we use the change of measure argument from \cite{junge2019stability} in order to directly connect the previous AT$(1,d)$ inequality from \cite{gao2017strong} for symmetric conditional expectations to an AT$(c,d')$ inequality for the general case, where $c$ is a spectral quantity depending solely on the invariant states of the smallest algebra $\cM$ and $d'$ is proportional to $d$. In particular, whenever $d=0$, this results allows us to transfer strong approximate tensorization for symmetric conditional expectations to strong approximate tensorization for general conditional expectations. However, in this inequality the multiplicative constant cannot be related to the clustering of correlations constant $c_1$ in the case of interacting systems, and can be in general exponentially larger. Our main result, stated in \Cref{theo_AT_pinching}, precisely fills this gap. Moreover, the inequality reduces to the classical inequality of \cite{cesi2001quasi} for commutative algebras. 

In Section \ref{sec:dabies}, we apply the previous results on weak approximate tensorization to the context of lattice spin systems with commuting Hamiltonians. In particular, we show in \Cref{prop_classical_inftyT} that classical evolutions over quantum systems (termed embedded Glauber dynamics) satisfy AT$(c,0)$ with the same constant as in the classical case. As an independent but important result, we also prove in \Cref{theo_equal_cond},  that the conditional expectations associated to the heat-bath dynamics and Davies dynamics coincide. This, in particular, allows us to transfer various results of remarkable interest that have been proven in the past years for one of the dynamics to the other, and vice versa.


\paragraph{Applications:} As mentioned previously, the main application of these inequalities is in the context of mixing times of continuous-time local Markovian evolutions over quantum lattice spin systems - although we expect these inequalities and their proof techniques to find other applications in quantum information theory. In \cite{cesi2001quasi}, Cesi used his inequality in order to show the exponential convergence in relative entropy of classical Glauber dynamics on lattice systems towards equilibrium, independently of the lattice size, in the form of a positive MLSI constant (defined in Section \ref{subsec:application_MLSI}). In a subsequent paper \cite{capel2020modified} that appeared after the first version of the current manuscript, we made use of the approximate tensorization inequality to show similar convergences for dissipative quantum Gibbs samplers.

Moreover, in this paper we illustrate the potential of these techniques in the aforementioned context of mixing times by estimating the MLSI constant whenever the generator of the dynamics is constructed from Pinching onto a pair of different, orthonormal bases. Additionally, we use our main results in approximate tensorization to obtain new entropic uncertainty relations in Section \ref{subsec:uncertainty_relations}.

\paragraph{Outline of the paper}
In \Cref{sec2}, we review basic mathematical concepts used in this paper, and more particularly the notion of a non-commutative conditional expectation. We derive theoretical expressions on the \textit{strong} ($c$) and \textit{weak} ($d$) constants for general von Neumann algebras in \Cref{sec:strong-quasi-tensorization}, where our main result is stated as \Cref{theo_AT_pinching}. We subsequently apply them to obtain strengthenings of uncertainty relations and examples of positivity of MLSI in \Cref{sec:applications}.  Moreover, in \Cref{sec:dabies}, we derive explicit bounds on the constants $c$ and $d$ for conditional expectations associated to Gibbs samplers on lattice spin systems in terms of the interactions of the corresponding Hamiltonian. In \Cref{sec:conclusion}, we discuss the results presented in our paper and how they have been applied to different contexts after the appearance of the first version of our manuscript.
Finally, in \Cref{sec:condexp}, we review the conditional expectations arising from Petz recovery maps and from Davies generators and show in that both conditional expectations coincide. We conclude by collecting the proofs of some technical results in Appendix \ref{appendix:proofs}.

\section{Notations and definitions}\label{sec2}
In this section, we fix the basic notation used in the paper, and introduce the necessary definitions. 

\subsection{Basic notations}

Let $(\cH,\langle .|.\rangle)$ be a finite-dimensional Hilbert space of dimension $d_\cH$. We denote by $\cB(\cH)$ the Banach space of bounded operators on $\cH$,  by $\cB_{\operatorname{sa}}(\cH)$ the subspace of self-adjoint operators on $\cH$, and by $\cB_+(\cH)$ the cone of positive semidefinite operators on $\cH$. The adjoint of an operator $Y$ is written as $Y^*$. We will also use the same notations $\cN_{\operatorname{sa}}$ and $\cN_+$ in the case of a von Neumann subalgebra $\cN$ of $\cB(\cH)$. The identity operator on $\cN$ is denoted by $\Id_\cN$, dropping the index $\cN$ when it is unnecessary. In the case of $\cB(\CC^\ell)$, $\ell\in\NN$, we will also use the notation $\Id$ for $\Id_{\CC^\ell}$. Similarly, given a map $\Phi:\cB(\cH)\to\cB(\cH)$, we denote its dual with respect to the Hilbert-Schmidt inner product as $\Phi_*$. We also denote by $\id_{\cB(\cH)}$, or simply $\id$, resp. $\id_\ell$, the identity superoperator on $\cB(\cH)$, resp. $\cB(\CC^\ell)$. We denote by $\mathcal{D}(\cH)$ the set of positive semidefinite, trace-one operators on $\cH$, also called \textit{density operators},  by $\cD_+(\cH)$ the subset of full-rank density operators, and by $\cD_{\leq}(\cH)$ the set of subnormalized density operators. In the following, we will often identify a density matrix $\rho\in\mathcal{D}(\cH)$ and the \textit{state} it defines, that is the positive linear functional $\cB(\cH)\ni X\mapsto\tr(\rho \,X)$. More generally, given a von Neumann subalgebra $\cN\subseteq\cB(\cH)$ with block decomposition $\cN:=\bigoplus_l\mathbb{M}_{n_l}\otimes \Id_{m_l}$, we denote by $\cD(\cN)$ the set of states of the form 
\begin{align}\label{eq:decompsigma}
\sigma:=\bigoplus_{l} p_l\,\rho_l\otimes\tau_l\,,
\end{align}
for some $n_l\times n_l$ states $\rho_l$ and $ m_l\times m_l$ full-rank states $\tau_l$. The sets $\cD(\cN)_+$ and $\cD(\cN)_{\le}$ are defined similarly.

\subsection{Entropic quantities and $\mathbb{L}_p$ spaces} Throughout this paper, we will use various distance measures between states and between observables: given a state $\rho\in\cD(\cN)$, its \textit{von Neuman entropy} is defined by
\begin{align*}
S(\rho):=-\tr\big[ \rho\,\ln\rho \big]\,.
\end{align*}
When $\rho\equiv \rho_{AB}\in \cD(\cH_A\otimes \cH_B)$ is the state of a bipartite quantum system, its \textit{conditional entropy} is defined by
\begin{align*}
S(A|B)_\rho:=S(\rho_{AB})-S(\rho_B)\,,
\end{align*}
where $\rho_B:=\tr_A(\rho)$ corresponds to the marginal of $\rho$ over the subsystem $\cH_B$. More generally, given two positive semidefinite operators $\rho,\sigma\in \cB_+(\cH)$, the \textit{relative entropy} between $\rho$ and $\sigma$ is defined as follows \cite{Umegaki-RelativeEntropy-1962}:
\begin{align*}
D(\rho\|\sigma):=\left\{\begin{aligned}
&\tr[\rho\,(\ln\rho-\ln\sigma)]\,\,\,\,\supp(\rho)\subset\supp(\sigma)\\
&+\infty\,\,\qquad\qquad\qquad\text{else}
\end{aligned}\right.
\end{align*}
Moreover, given (possibly subnormalized) positive semidefinite operators $\rho\ge 0$ and $\sigma>0$, their \textit{max-relative entropy} is defined as \cite{datta2009minmaxrelativeentropies}: 
\begin{align*}
D_{\max}(\rho\|\sigma):= \inf\{\lambda|\,\rho\le \e^{\lambda}\sigma    \}  \equiv\ln\,(\|\sigma^{-\frac{1}{2}}\,\rho\,\sigma^{-\frac{1}{2}}\|_\infty)\,.
\end{align*}
From the max-relative entropy, we can define the \textit{max-information} of a (possibly subnormalized) bipartite state $\rho_{AB}\in\cD_{\le}(\cH_A\otimes \cH_B)$ as follows \cite{berta2011quantum}:
\begin{align*}
I_{\max}(A:B)_{\rho}\equiv I_{\max}(\cH_A:\cH_B)_{\rho} :=\inf_{\tau_{B}\in\cD(\cH)}\,D_{\max}(\rho_{AB}\|\rho_A\otimes \tau_B)\,.
\end{align*}
Given a subalgebra $\cN$ of $\cB(\cH)$ and $\sigma\in\cD_+(\cN)$, we define the modular maps $\Gamma_\sigma:\cN\to \cB(\cH)$ and $\Delta_\sigma:\cN\to\cN$ as follows
\begin{align*}
	\Gamma_\sigma(X):=\sigma^{1/2}\,X\,\sigma^{1/2}\, , \qquad\Delta_\sigma(X)=\sigma\,X\,\sigma^{-1}\,.
	\end{align*}
Then for any $p\ge 1$ and $X\in\cN$, its non-commutative weighted $\mathbb{L}_p(\sigma)$-norm is defined as \cite{Kosaki-noncommLp-1984}:
\begin{align*}
	\|X\|_{\mathbb{L}_p(\sigma)}:=\tr\left[ \big|\Gamma_\sigma^{\frac{1}{p}}(X) \big|^p  \right]^{\frac{1}{p}}\, ,
	\end{align*}
	and $\|X\|_{\mathbb{L}_{\infty}(\sigma)}=\|X\|_\infty$, the operator norm of $X$, which we will also often more simply denote by $\| X \|$. We call the space $\cB(\cH)$ endowed with the norm $\|.\|_{\mathbb{L}_p(\sigma)}$ the \textit{quantum $\mathbb{L}_p(\sigma)$ space}. In the case $p=2$, we have a Hilbert space, with corresponding $\sigma$-KMS scalar product
\begin{align}\label{KMSinner}
	\langle X,\,Y\rangle_\sigma:=\tr \left[\sigma^{1/2}X^*\sigma^{1/2}Y \right]\,.
	\end{align}
Weighted $\mathbb{L}_p$ norms enjoy the following useful properties:
\begin{itemize}
	\item[-] H\"{o}lder's inequality: for any $p,\hat{p}\ge 1$ such that $p^{-1}+\hat{p}^{-1}=1$, and any $X,Y\in\cN$:
	\begin{align*}
		\langle X,\,Y\rangle_\sigma\le \|X\|_{\mathbb{L}_p(\sigma)}\,\|Y\|_{\mathbb{L}_{\hat{p}}(\sigma)}\,.
		\end{align*}
		Here, $\hat{p}$ is the \textit{H\"{o}lder conjugate} of $p$.
	\item[-] Duality of norms: for any $p\ge 1$ of H\"{o}lder conjugate $\hat{p}$, and any $X\in\cN$:
	\begin{align*}
		\|X\|_{\mathbb{L}_p(\sigma)}=\sup_{\|Y\|_{\mathbb{L}_{\hat{p}}(\sigma)}\le 1}\,\langle Y,\,X\rangle_\sigma\,.
		\end{align*}
	\item[-] For any completely positive, unital linear map $\Phi:\cN\to \cN$ such that $\Phi_*(\sigma)=\sigma$, any $p\ge 1$ and any $X\in\cN$:
	\begin{align}\label{Phiinvariantnorm}
		\|\Phi(X)\|_{\mathbb{L}_p(\sigma)}\le \|X\|_{\mathbb{L}_p(\sigma)}\,.
		\end{align}
	\end{itemize}
	
\subsection{Conditional expectations}\label{subsec:conditional_expectations}

Here, we introduce the main object studied in this paper:

\begin{definition}[Conditional expectations~\cite{OhyaPetz-Entropy-1993}]
 Let $\cM\subset \cN$ be a von Neumann subalgebra of $\cN$. Given a state $\sigma\in\cD_+(\cM)$, a linear map $E:\cN\to\cM$ is called a \textit{conditional expectation} with respect to $\sigma$ of $\cN$ onto $\cM$ if the following conditions are satisfied:
\begin{itemize}
	\item[-] For all $X\in\cN$, $\|E[X]\|\le \|X\|$;
	\item[-] For all $X\in \cM$, $E[X]=X$;
	\item[-] For all $X\in\cN$, $\tr[\sigma E[X]]=\tr[\sigma X]$.
\end{itemize}
\end{definition}

A conditional expectation satisfies the following useful properties (see \cite{Aspects2003} for proofs and more details):

\begin{proposition}\label{propositioncondexp}
	Conditional expectations generically satisfy the following properties:
	\begin{itemize}
		\item[(i)] The map $E$ is completely positive and unital.
		\item[(ii)] For any $X\in\cN$ and any $Y,Z\in\cM$, $E[YXZ]=Y E[X]Z$.
		\item[(iii)] $E$ is self-adjoint with respect to the scalar product $\langle .,\,.\rangle_\sigma$. In other words:
		\begin{align*}
			\Gamma_\sigma\circ E=E_*\circ\Gamma_\sigma   \,,
		\end{align*}
		where $E_*$ denotes the adjoint of $E$ with respect to the Hilbert-Schmidt inner product.
		\item[(iv)] $E$ commutes with the modular automorphism group of $\sigma$: for any $s\in\RR$,
		\begin{align}
			\Delta_\sigma^{is}\circ E=E\circ \Delta^{is}_\sigma\,.
		\end{align}
		\item[(v)] Uniqueness: given a von Neumann subalgebra $\cM\subset \cN$ and a faithful state $\sigma$, the existence of a conditional expectation $E$ is equivalent to the invariance of $\cM$ under the modular automorphism group $(\Delta_\sigma^{is})_{s\in\RR}$.  In this case, $E$ is uniquely determined by $\sigma$.
	\end{itemize}
\end{proposition}

From now on, and with a slight abuse of notations, given the finite-dimensional von Neumann subalgebra $\cN=E[\cB(\cH)]$ of $\cB(\cH)$, we denote by $\cD(\cN):= E_{*}(\cD(\cH))$ its corresponding set of states that are invariant by $E$, so that $\cD(\cH)\equiv \cD(\cB(\cH))$. In other words, the states $\tau_l$ in the decomposition \eqref{eq:decompsigma} are now fixed by $E$. Similarly, the set of subnormalized states on the algebra $\cN$ is defined as $\cD_{\le}(\cN) $. We also introduce the concept of a conditional covariance: given a von Neumann-subalgebra $\cM\subset \cN$, a conditional expectation $E^\cM$ from $\cN$ onto $\cM$ and a quantum state $\sigma\in\cD_+(\cM)$, where $\cD(\cM)$ is defined with respect to $E^\cM$, we define the \textit{conditional covariance} functional as follows: for any two $X,Y\in\cN$,
\begin{align}\label{condcov}
	\operatorname{Cov}_{\cM,\sigma}(X,Y):=\langle X-E^\cM[X],\,Y-E^\cM[Y]\rangle_\sigma\,.
\end{align}

\subsection{Two examples of classes of conditional expectations}\label{sect_twoexamples}

In this subsection, we provide more details about the conditional expectations that we will consider in the case of Gibbs states on lattice spin systems in Section \ref{sec:dabies}. Some properties and new results of independent interest regarding these conditional expectations are deferred to Appendix \ref{sec:condexp} for sake of clarity.

\subsubsection{Conditional expectations generated by a Petz recovery map}

Let $\sigma$ be a faithful density matrix on a finite-dimensional algebra $\cN$ and let $\cM\subset\cN$ be a subalgebra. We denote by $E_\tau$  the conditional expectation onto $\cM$ with respect to the completely mixed state (i.e. $E_\tau$ is self-adjoint with respect to the Hilbert-Schmidt inner product). We also adopt the following notations: we write $\sigma_\cM=E_\tau(\sigma)$ and
\[\mathcal A_\sigma(X):=\sigma_\cM^{-\frac12}\,E_\tau[\sigma^{\frac12}\,X\,\sigma^{\frac12}]\,\sigma_\cM^{-\frac12}\,.\]
Remark that $\mathcal A_\sigma$ is also the unique map such that for all $X\in\cN$ and all $Y\in\cM$:
\[\tr[\sigma^{\frac12}\,X\,\sigma^{\frac12}\,Y]=\tr[\sigma_\cM^{\frac12}\,\mathcal A_\sigma(X)\,\sigma_\cM^{\frac12}\,Y]\,.\]
The adjoint of $\mathcal A_\sigma$ is the Petz recovery map of $E_\tau$ with respect to $\sigma$, denoted by $\cR_\sigma$:
\begin{align*}
\mathcal{R}_{\sigma}(\rho_\cM):=\sigma^{\frac{1}{2}}\sigma_\cM^{-\frac{1}{2}}\rho_\cM\sigma_\cM^{-\frac{1}{2}}\sigma^{\frac{1}{2}}\,,
\end{align*}
where $\rho_\cM:=E_\tau(\rho)$. It is proved in \cite{carlen2017recovery} that $\mathcal A_\sigma$ is a conditional expectation if and only if $\sigma\,X\,\sigma^{-1}\in\cM$ for all $X\in\cM$. In the general case, we denote by 
\begin{equation}\label{eq_cond_Petz}
E_\sigma:=\lim_{n\to\infty}\mathcal A_\sigma^n 
\end{equation}
the projection on its fixed-point algebra for the $\sigma$-KMS inner product, which is a conditional expectation as we assumed $\sigma$ to be faithful. That is, $E_\sigma$ is the orthogonal projection for $\langle\cdot,\cdot\rangle_\sigma$ on the algebra:
\[\cF(\mathcal A_\sigma)=\{X\in\cN\,;\,\mathcal A_\sigma(X)=X\}\,.\]

	\subsubsection{Conditional expectations coming from Davies semigroups}
	
The basic model for the evolution of an open system in the Markovian regime is given by a quantum Markov semigroup (or QMS) $(\cP_t)_{t\ge0}$ acting on $\cB(\cH)$. Such a semigroup is characterised by its generator, called the Lindbladian $\LL$, which is defined on $\cB(\cH)$ by 
$$\cL(X)={\lim}_{t\to 0}\,\frac{1}{t}\,(\cP_t(X)-X) $$
 for all $X\in\cB(\cH)$. Recall that by the GKLS Theorem \cite{Lind,[GKS76]}, $\cL$ takes the following form: for all $X\in\cB(\cH)$,
\begin{equation}\label{eqlindblad}
\cL(X)=i[H,X]+\frac{1}{2}\sum_{k=1}^l{\left[2\,L_k^*XL_k-\left(L_k^*L_k\,X+X\,L_k^*L_k\right)\right]}\, , 
\end{equation}
where $H\in\cB_{\operatorname{sa}}(\cH)$, the sum runs over a finite number of \textit{Lindblad operators} $L_k\in\cB(\cH)$, and $[\cdot,\cdot]$ denotes the commutator defined as $[X,Y]:=XY-YX$, $\forall X,Y\in\cB(\cH)$. The QMS is said to be \textit{faithful} if it admits a full-rank invariant state $\sigma$. When the state $\sigma$ is the unique invariant state, the semigroup is called \textit{primitive}. Further assuming the self-adjointness of the generator $\cL$ with respect to the inner product (\ref{KMSinner}) (or \textit{detailed balance condition}), there exists a conditional expectation $E\equiv E_\cF$ onto the fixed-point subalgebra $\cF(\cL):=\{X\in\cB(\cH):\,\cL(X)=0\}$ such that
\begin{align*}
\cP_t(X)\underset{t\to\infty}{\to}E[X]\,
\end{align*}
 for all $X\in\cB(\cH)$. \\
 
 We now focus on a particular class of QMS called Davies QMS. Such semigroups are obtained in the \textit{weak coupling limit} of a system and a heat bath.
 Let $H$ be a selfadjoint operator on $\cH$, representing the Hamiltonian of the system. The corresponding Gibbs state at inverse temperature $\beta$ is defined as
\begin{align}
\sigma=\frac{\e^{-\beta H}}{\tr[\e^{-\beta H}]}\,.
\end{align}

Next, consider the Hamiltonian $H^{\operatorname{HB}}$ of the heat bath, as well as a set of system-bath interactions $\{ S_{\alpha}\otimes B_{\alpha} \}$, for some label $\alpha$. Here, we do not assume anything on the $S_\alpha$'s. The Hamiltonian of the universe composed of the system and its heat-bath is given by
\begin{align}
H=H_\Lambda+H^{\operatorname{HB}}+\sum_{\alpha\in\Lambda}S_{\alpha}\otimes B_{\alpha}\,.
\end{align}
Assuming that the bath is in a Gibbs state, by a standard argument (e.g. weak coupling limit, see \cite{[SL78]}), the evolution on the system can be approximated by a quantum Markov semigroup whose generator is of the following form:
\begin{align}\label{eq_lindblad}
\cL^{\operatorname{D},\beta}(X)=\sum_{\omega,\alpha}\,\chi^{\beta}_{\alpha}(\omega)\,\Big(  S_{\alpha}^*(\omega)XS_{\alpha}(\omega)-\frac{1}{2}\,\big\{  S_{\alpha}^*(\omega)S_{\alpha}(\omega),X \big\}   \Big)\,.
\end{align}
The Fourier coefficients of the two-point correlation functions of the environment $\chi_{\alpha}^\beta$ satisfy the following KMS condition:
\begin{align}\label{eq_KMS}
\chi_{\alpha}^\beta(-\omega)=\e^{-\beta\omega}\,\chi_{\alpha}^\beta(\omega)\,.
\end{align}
The operators $S_{\alpha}(\omega)$ are the Fourier coefficients of the system couplings $S_{\alpha}$, which means that they satisfy the following equation for any $t\in\RR$:
\begin{align}\label{eq!}
\e^{-itH}\,S_{\alpha}\e^{it H}=\sum_\omega\e^{it\omega}S_{\alpha}(\omega)\,\qquad\Leftrightarrow  \qquad  S_\alpha(\omega)=\sum_{\eps-\eps'=\omega}P_\eps\,S_\alpha\,P_{\eps'}\,.
\end{align}
where the sum is over a finite number of frequencies. This implies in particular the following useful relation:
\begin{align}\label{eq_eigenvector1}
\Delta_{\sigma}(S_{\alpha}(\omega))=\e^{\beta\omega}\,S_{\alpha}(\omega)\,.
\end{align}
The above identity means that the operators $S_{\alpha}(\omega)$ form a basis of eigenvectors of $\Delta_\sigma$. Next, we define the conditional expectation onto the algebra $\cF(\cL)$ of fixed points of $\cL$ with respect to the Gibbs state $\sigma=\sigma^\beta$ as follows \cite{Kastoryano2014}: 
\begin{align}\label{Davies_cond}
E^{\operatorname{D},\beta}:=\lim_{t\to \infty}\e^{t\cL^{\operatorname{D},\beta}}\,.
\end{align}

Some results regarding the fixed-point algebra associated to this conditional expectation are contained in Appendix \ref{sec:condexp}. In particular, we prove the following theorem which is of independent interest.

\begin{theorem}\label{theo_equal_cond}
 Define the algebra $\cM=\{S_\alpha\}'$, $E^{\operatorname{D},\beta}$ as above and $E_{\sigma}$ as in \Cref{eq_cond_Petz} with respect to the inclusion $\cM\subset\cB(\cH)$. Then both conditional expectations coincide.
\end{theorem}

\section{Weak approximate tensorization of the relative entropy}\label{sec:strong-quasi-tensorization}

This section is devoted to the main results of this article, namely approximate tensorization inequalities for the relative entropy. 

\begin{definition}
	Let $\cM\subset \cN_1,\,\cN_2\subset \cN$ be finite-dimensional von Neumann algebras and $E^{\cM},\,E_1 ,\, E_2$ associated conditional expectations onto $\cM$, resp. $\cN_1,\,\cN_2$. These conditional expectations are said to satisfy a \textit{weak approximate tensorization} with constants $c \geq1$ and $d\geq 0$, denoted by AT$(c,d)$, if, for any state $\rho\in\cD(\cN)$:
	\begin{align}\tag{$\operatorname{AT(c,d)}$}\label{at}
	D(\rho\|E^\cM_*(\rho))\le c\,\big(D(\rho\|E_{1*}(\rho))+D(\rho\|E_{2*}(\rho))\big)+d
	\end{align}
	The approximate tensorization is said to be \textit{strong} if $d=0$.
\end{definition}

\begin{remark}
	One can easily get similar inequalities for $k\ge 2$ algebras $\cM\subset \cN_1,\dots \cN_k\subset \cN$ by simply averaging over each inequality for two $k_1\ne k_2\in[k]$. Denoting by $c$ and $d$ as the maximal constants we get by considering two algebras $\cN_{k_1}$ and $\cN_{k_2}$ pairwise, we would thus obtain
	\begin{align}\label{at_multiple_k}
	D(\rho\|E^\cM_*(\rho))\le \frac{2c}{k}\,\sum_{j=1}^k\,D(\rho\|E_{j*}(\rho))+d\,.
	\end{align}
	For sake of clarity, we will restrict to the case $k=2$ in the rest of the article.
	\end{remark}

The first technical result presented in this section is Lemma \ref{propAT}, derived from the so-called \textit{multivariate trace inequalities} \cite{Sutter2017}. It takes the form 
	\begin{align*}
D(\rho\|E^\cM_*(\rho))\le D(\rho\|E_{1*}(\rho))+D(\rho\|E_{2*}(\rho))  + \xi (E_{1*}(\rho), \,  E_{2*}(\rho), \, E^\cM_*(\rho))\, ,
\end{align*}
where $\xi (E_{1*}(\rho), \,  E_{2*}(\rho), \, E^\cM_*(\rho))$ is an additive error term that we subsequently estimate via different approaches in the subsequent \Cref{sec:changemeas,subsec:ApproxTensor_Pinching,sec:clustering}: Lemma \ref{propAT} directly yields a generalization of a result of \cite{gao2017strong} for conditional expectations with respect to non-tracial states in Corollary \ref{corollaruweak}. Moreover, using a noncommutative change of measure argument \cite{bardet2021group}, we provide in \Cref{HSAT} some first estimates of the strong and weak constants $c$ and $d$ in AT($c,d$) in terms of the maximal and minimal eigenvalues of a common invariant state of the three conditional expectations involved. 

Next, in \Cref{theo_AT_pinching}, we use a different technique involving Pinching maps onto certain subspaces that appear in a block-diagonal decomposition of $\mathcal{M}$ (this setting is properly introduced in Section \ref{subsec:ApproxTensor_Pinching}) to obtain the inequality:
	\begin{align*}
D(\rho\|E^\cM_*(\rho))\le \frac{1}{1-c_1} \left( D(\rho\|E_{1*}(\rho))+D(\rho\|E_{2*}(\rho))  \right) + \xi_2 (E_{1*}(\rho), \,  E_{2*}(\rho), \, E^\cM_*(\rho))\, ,
\end{align*}
 where $\xi_2 (E_{1*}(\rho), \,  E_{2*}(\rho), \, E^\cM_*(\rho))$ strongly depends on the Pinching map with respect to  $E^\cM_*(\rho)$ and it is subsequently estimated in Proposition \ref{propboundsd1d2}. Furthermore, the multiplicative error term above can be interpreted as arising from a condition of clustering of correlations for the state $E_*^\cM(\rho)$ (see Section \ref{sec:clustering}). 

\subsection{A technical lemma}

In the next result, we derive a bound on the difference between $D(\rho\|E^\cM_*(\rho))$ and the sum of the relative entropies $D(\rho\|E_{i*}(\rho))$, which is our key tool in finding constants $c$ and $d$ for which \ref{at} is satisfied. The result is inspired by the work of \cite{cesi2001quasi,[D02]} and makes use of the multivariate trace inequalities introduced in \cite{Sutter2017}:

\begin{lemma}\label{propAT}
	Let $\cM\subset \cN_1 ,\,\cN_2\subset \cN$ be finite-dimensional von Neumann algebras and $E^{\cM},E_1 ,\, E_2$ their corresponding conditional expectations. Then the following inequality holds for any $\rho\in\cD(\cN)$, writing $\rho_j:=E_{j*}(\rho)$ and $\rho_\cM:=E^\cM_{*}(\rho)$:
	\begin{align}\label{mainequation}
	D(\rho\|\rho_\cM)\le \,D(\rho\|\rho_1)+D(\rho\|\rho_2)+\,\ln \left\lbrace \int_{-\infty}^\infty\,\tr\left[\rho_{1}\, \rho_{\cM}^{\frac{-1-it}{2}}\, \rho_{2}\, \rho_{\cM}^{\frac{-1+it}{2}}   \right]\,\beta_0(t)\,dt \right\rbrace \,,
	\end{align}
	with the probability distribution function
	\begin{equation*}
	\beta_0(t)= \frac{\pi}{2} (\cosh(\pi t)+ 1)^{-1}\,.
	\end{equation*}
\end{lemma}

\begin{proof}
	 The first step of the proof consists in showing the following bound:
	\begin{equation}\label{eq:step-31}
	D(\rho\|\rho_\cM) \leq D(\rho\|\rho_1)+D(\rho\|\rho_2)+ \ln\tr [M]\,,
	\end{equation}
	where
	$  M = \exp \left[ - \ln \rho_\cM + \ln \rho_1  + \ln \rho_2 \right] $.
	Indeed, 
	\begin{align*}
	\ds D(\rho\|\rho_\cM)- D(\rho\|\rho_1)- D(\rho\|\rho_2) &=\tr \left[ {\rho} \left( - \ln {\rho} \underbrace{ - \ln \rho_\cM + \ln \rho_1+\ln \rho_2}_{\ln M} \right) \right] \\
	& = - D(\rho \| M).
	\end{align*}
	Moreover, since $\tr[M]\neq 1$ in general, from the non-negativity of the relative entropy of two states it follows that:
	\begin{equation*}
	D(\rho \| M) \geq  -\log \tr[M].
	\end{equation*}
	
	In the next step, we bound the error term making use of \cite[Theorem 7]{Lieb1973}  and \cite[Lemma 3.4]{Sutter2017}, concerning Lieb's extension of Golden-Thompson inequality and Sutter, Berta and Tomamichel's rotated expression for Lieb's pseudo-inversion operator using multivariate trace inequalities, respectively: Let us recall that Theorem 7 of \cite{Lieb1973} states that for observables $f, g$ and $h$, we have
	\begin{align*}
 \tr \left[ \operatorname{exp}\left( - f + g + h  \right) \right] \leq  \tr \left[ \operatorname{e}^g \mathcal T_{\operatorname{e}^f} (\operatorname{e}^h)\right],
	\end{align*}
	where $\mathcal T_{f} $ is given by:
	\begin{equation*}
	\mathcal T_{f} (h) := \int_0^\infty  (f + t)^{-1} h (f + t)^{-1} dt \, .
	\end{equation*}
	An alternative definition of this superoperator in terms of multivariate trace inequalities was provided in Lemma 3.4 of \cite{Sutter2017}, namely
		\begin{equation*}
	\mathcal T_{f} (h) = \int_{- \infty}^\infty  \beta_0 (t) \, f^{\frac{-1-it}{2}} h f^{\frac{-1+it}{2}} dt \, ,
	\end{equation*}
		with $\beta_0$ as in the statement of the lemma. Now, we apply both results to inequality (\ref{eq:step-31}), to obtain
	\begin{align*}
	\tr [M] =& \tr \left[ \operatorname{exp}\left( - \ln  {\rho_\cM} +
	\ln {\rho_1} +
	\ln {\rho_2}  \right) \right] \leq  \int_{- \infty}^\infty \tr \left[ \rho_1  \, \rho_\cM^{\frac{-1-it}{2}}  \rho_2 \, \rho_\cM^{\frac{-1+it}{2}} \right]\,\beta_0 (t) \, dt\,,
	\end{align*}
which concludes the proof of the lemma.

\end{proof}

Note that, if a constant $d>0$  is such that 
\begin{equation*}
\ln \int_{-\infty}^\infty\,\tr\left[\rho_{1}\rho_{\cM}^{\frac{-1-it}{2}}\rho_{2}\rho_{\cM}^{\frac{-1+it}{2}}   \right]\,\beta_0(t)\,dt \leq d
\end{equation*}
for every $\rho \in \cD(\cN)$, then inequality \eqref{mainequation} constitutes a result of approximate tensorization AT($1,d$). Using this observation, we obtain an arguably more direct proof of a result appearing in \cite{gao2017strong}, that we generalize to the case of non-tracial states. Indeed, the proof of \cite{gao2017strong} required the introduction of so-called amalgamated $\mathbb{L}_p$ spaces, a technical tool that we do not require. 
\begin{corollary}\label{corollaruweak}
	With the notations of \Cref{propAT}, define the constant $$d:=\sup_{\rho\in\cD(\cN_{2})}\inf\big\{\ln(\lambda)|\,E_{1*}(\rho)\le \lambda\eta\,\text{ for some }\,\eta\in\cD(\cM)  \big\}\equiv \sup_{\rho\in\cD(\cN_{2})}\inf_{\eta\in\cD(\cM)}\,D_{\max}(E_{1*}(\rho)\|\eta)  \,.$$ Then the following weak approximate tensorization $\operatorname{AT}(1,d)$ holds:
	\begin{align*}
	D(\rho\|\rho_\cM)\le D(\rho\|\rho_1)+D(\rho\|\rho_2)+ d\,.
	\end{align*}
\end{corollary}
\begin{proof} 
	We focus on the last term on the right-hand side of (\ref{mainequation}). First, note that:
\begin{equation*}
	\tr \left[ \rho_{1}\,\rho_\cM^{\frac{-1-it}{2}}\rho_{2}\,\rho_\cM^{\frac{-1+it}{2}} \right] = \tr \left[ \rho\,E_1^*\left( \rho_\cM^{\frac{-1-it}{2}}\rho_{2}\,\rho_\cM^{\frac{-1+it}{2}} \right) \right] = \tr \left[ \rho\,\rho_\cM^{\frac{-1-it}{2}}E_{1*}(\rho_{2})\,\rho_\cM^{\frac{-1+it}{2}} \right].
	\end{equation*}	
	 We have by definition of $d$ that there exists a state $\eta\in\cD(\cM)$ such that for any $t\in\RR$:
	\begin{align*}
	\tr \left[ \rho\,\rho_\cM^{\frac{-1-it}{2}}E_{1*}(\rho_{2})\,\rho_\cM^{\frac{-1+it}{2}} \right] \le\,\e^{d} \tr \left[ \rho\,\rho_\cM^{\frac{-1-it}{2}}\eta\,\rho_\cM^{\frac{-1+it}{2}} \right]=\e^{d}\tr[\rho X_\cM]\,,
	\end{align*}
	for some density $X_\cM\in\cM$ given by $\rho_\cM^{\frac{-1-it}{2}}\eta\,\rho_\cM^{\frac{-1+it}{2}}$. Since $\cM\subset \cN$, $\tr[\rho X_\cM]=\tr[\rho_\cM\,X_\cM]=\tr[\eta]=1$. The result follows.
\end{proof}

\begin{remark}
	In \cite{gao2019relative}, the authors showed that, for doubly stochastic conditional expectations (i.e. $E_{i*}=E_i$, $E^\cM_*=E^\cM$), the following equation holds: Given the following block decomposition of the algebras $\cN_2$ and $\cM$,
	\begin{align*}
	\cN_2\equiv \bigoplus_{l\in I_{\cN_2}}\,\mathbb{M}_{m_l}\otimes \Id_{t_l}\,\qquad\qquad	\cM\equiv \bigoplus_{k\in I_{\cM}}\,\mathbb{M}_{n_k}\otimes \Id_{s_k}\,,
	\end{align*} 
	\begin{align*}
	D(\cN_2\|\cM):= \sup_{\rho\in\cD(\cN_2)} \inf_{\eta\in\cD(\cM)}D_{\max}(\rho\| \eta)\equiv \max_{l\in I_{\cN_2}}\ln \Big(\sum_{k\in I_{\cM}}  \min(a_{kl},\,n_k)\,s_k/t_l \Big)\,,
	\end{align*}
	where $a_{kl}$ denotes the number of copies of the block $\mathbb{M}_{n_k}$ contained in the block $\mathbb{M}_{m_l}$. In the context of lattice spin systems, this typically corresponds to 
	the	infinite temperature regime.
\end{remark} 

\subsection{Approximate tensorization via noncommutative change of measure}\label{sec:changemeas}

\Cref{corollaruweak} states a correction to exact tensorization with a unique weak constant. We expect this result to be relevant for doubly stochastic conditional expectations, where this additive term is purely quantum. However, the weak constant $d$ is suboptimal in general. In this section and the following one, we provide tools to improve the latter at the cost of replacing the optimal strong constant by $c>1$. This intuition is inspired by the classical setting, where the weak constant can be removed at the cost of a worsening of the strong constant \cite{cesi2001quasi,[D02]}.

Given a state $\sigma$ that is invariant for the conditional expectations $E^\cM, E_1$ and $E_2$, we define the doubly stochastic conditional expectations ${E}^{(0),\cM}, E_1^{(0)}$ and $E_2^{(0)}$ onto the same fixed-point algebras $\cM\subset \cN_1,\cN_2\subset \cN $. Then, the following proposition is a direct consequence of a recent noncommutative change of measure argument in  \cite{junge2019stability} under the assumption that strong approximate tensorization for the relative entropy holds for ${E}^{(0),\cM}, E_1^{(0)}$ and $E_2^{(0)}$. 

\begin{proposition}\label{HSAT}
	As in \Cref{corollaruweak}, we define the constant 
	\[d:=\sup_{\rho\in\cD(\cN_{2})}\inf\big\{\ln(\lambda)|\,E_{1*}^{(0)}(\rho)\le \lambda\eta\,\text{ for some }\,\eta\in\cD(\cM)  \big\}\equiv \sup_{\rho\in\cD(\cN_{2})}\inf_{\eta\in\cD(\cM)}\,D_{\max}(E^{(0)}_{1*}(\rho)\|\eta)  \,.\]
Let us assume that $\operatorname{AT}(1,d)$  holds for the doubly stochastic conditional expectations, i.e. for every $\rho \in \cD(\cH)$
\begin{equation}\label{eq:InfiniteTemperatureAT}
D( \rho \|E^{(0),\cM}_{*}(\rho)) \leq D(\rho\|E^{(0)}_{1*}(\rho)) + D(\rho \|E^{(0)}_{2*}(\rho)) +d \, .
\end{equation}
	Then, the following result of $\operatorname{AT}(c,d')$  with $c=\frac{\lambda_{\max}(\sigma)}{\lambda_{\min}(\sigma)}$ and $d'= \lambda_{\max}(\sigma)\,d_\cH\,d$ holds:
	\begin{align}\label{eq:NoncommChangMeasAT}
	D(\rho\|E^\cM_{*}(\rho))\le \frac{\lambda_{\max}(\sigma)}{\lambda_{\min}(\sigma)}\,\big(D(\rho\|E_{1*}(\rho))+D(\rho\|E_{2*}(\rho))\big)+\lambda_{\max}(\sigma)\,d_\cH\,d\,.
	\end{align}
In particular, if $\operatorname{AT}(1,0)$  holds for the doubly stochastic conditional expectations ${E}^{(0),\cM}, E_1^{(0)}$ and $E_2^{(0)}$, then the conditional expectations $E^\cM, E_1$ and $E_2$ satisfy $\operatorname{AT}(c,0)$  with $c=\frac{\lambda_{\max}(\sigma)}{\lambda_{\min}(\sigma)}$.
\end{proposition} 

We defer the proof of this result to the Appendix \ref{subsec:proof_change_measure}, as it merely follows the lines of \cite{junge2019stability}.

\subsection{Approximate tensorization via Pinching map}\label{subsec:ApproxTensor_Pinching}
\Cref{HSAT} states an approximate tensorization inequality with the advantage over 
\Cref{corollaruweak} that the weak constant $d$ vanishes when the doubly stochastic conditional expectations projecting onto the same subalgebras form a commuting square. However, the multiplicative constant typically explodes when increasing the size of the system. In the following theorem, we take care of this issue by employing a pinching argument in place of the change of measure argument laid in \Cref{HSAT}. 


Before stating the result, let us fix some notations. As before, we are interested in proving (weak) approximate tensorization results for the quadruple of algebras $\cM\subset \cN_1\,,\,\cN_2\subset\cn$. As a subalgebra of $\cB(\cH)$ for some Hilbert space $\cH$, $\cM$ bears the following block diagonal decomposition: given $\cH=\bigoplus_{i\in I_\cM}\cH_i\otimes \cK_i$:
	\begin{align}\label{eq_decomp}
	\cM\equiv\bigoplus_{i\in I_\cM}\,\cB(\cH_i)\otimes \Id_{\cK_i}\,,\qquad\text{ so that }\qquad\forall \rho\in\cD(\cN)\,,\, \rho_\cM:=\sum_{i\in I_\cM}\,\tr_{\cK_i}[P_i\rho P_i]\,\otimes\tau_i\,,
	\end{align}
where $P_i$ corresponds to the projection onto the $i$-th diagonal block in the decomposition of $\cM$, and each $\tau_i$ is a full-rank state on $\cK_i$. We further make the observation that, since the restrictions of the conditional expectations $E_1$, $E_2$ and $E^\cM$ on $\cB(\cH_i\otimes\cK_i)$ only act non-trivially on the factor $\cB(\cK_i)$, there exist conditional expectations ${E}_j^{(i)}$ and $({E}^{\cM})^{(i)}$ acting on $\cB( \cK_i)$ and such that 
\begin{equation}\label{eq_decom_cond}
 E_j|_{\cB(\cH_i\otimes\cK_i)}:=\id_{\cB(\cH_i)}\otimes {E}_j^{(i)}\,,\quad\text{resp.}\quad E^\cM|_{\cB(\cH_i\otimes\cK_i)}:=\id_{\cB(\cH_i)}\otimes ({E}^\cM)^{(i)}\,.
\end{equation}
In order to get another form of approximate tensorization, we wish to compare the state $\rho$ with a classical-quantum state according to the decomposition given by $\cM$. To this end we introduce the Pinching map with respect to each $\cH_i$: define $\rho_{\cH_i}\equiv \tr_{\cK_i}[P_i\,\rho\,P_i]$. Then each $\rho_{\cH_i}$ can be diagonalized individually:
\[\rho_{\cH_i}\equiv\sum_{\lambda^{(i)}\in\operatorname{Spec}(\rho_{\cH_i})}\,\lambda^{(i)}\,\proj{\lambda^{(i)}}\,.\]
The Pinching map we are interested in is then:
\[\cP_{\rho_\cM}(X)\equiv \sum_{i\in I_\cM}\,\sum_{\lambda^{(i)}\in\operatorname{Spec}(\rho_{\cH_i})}\,\left(\proj{\lambda^{(i)}}\otimes \ind_{\cK_i}\right)\,X\,\left(\proj{\lambda^{(i)}}\otimes \ind_{\cK_i}\right)\,,\qquad X\in\cB(\cH)\,.\]
Remark that we have for all $\rho\in\cD(\cN)$:
\[\tr_{\cH_i}[P_i\,\rho\,P_i]=\tr_{\cH_i}[P_i\,\cP_{\rho_\cM}(\rho)\,P_i]\,.\]

\begin{theorem}\label{theo_AT_pinching}
 Assume 
	\begin{align}\label{cond_L1_clustering}
	&c_1:=\max_{i\in I_\cM}\|E_{1}^{(i)}\circ E_{2}^{(i)}-(E^\cM)^{(i)}:\,\mathbb{L}_1(\tau_i)\to\mathbb{L}_\infty\|<1\,.
	\end{align}
	Then, the following inequality holds:
	\begin{align}\label{eqgeneral}
	D(\rho\|\rho_\cM)\le \frac{1}{(1-{c_1})}\,\big(D(\rho\|\rho_1)+D(\rho\|\rho_2)+D_{\max}\big(E_{1*}\circ E_{2*}(\rho)\|E_{1*}\circ E_{2*}(\eta)\big)+\,c_1 D(\eta\|\cP_{\rho_\cM}(\rho))\big)\,,
	\end{align}
	for any $\eta\in\cD(\cN)$ such that $\eta=\cP_{\rho_\cM}(\eta)$ and $\tr_{\cK_i}[P_i\,\eta\,P_i]=\rho_{\cH_i}$.
	In particular, any state $\eta$ of the form $\eta:= \sum_{i\in I_\cM}\,\rho_{\cH_i}\otimes \tau_i'$, for an arbitrary family of subnormalized states $\tau_i'$, satisfies these conditions.

	Alternatively, we can get 
	\begin{align}\label{eqchanged}
	D(\rho\|\rho_\cM)\le \frac{1}{(1-{c_1})}\,\big(D(\rho\|\rho_1)+D(\rho\|\rho_2)\big)+D\big(\rho\|\cP_{\rho_\cM}(\rho) )\,.
	\end{align}
	Consequently, \ref{at} holds with 
	\begin{align}
	 & c:=\frac{1}{(1-{c_1})}\,, \nonumber\\
	 & d:= \frac{1}{(1-{c_1})}\left(\underset{\rho\in\cD(\cN)}{\sup}\,\underset{\eta\in\cD(\cN)}\inf\,D_{\max}\big(E_{1*}\circ E_{2*}(\rho)\|E_{1*}\circ E_{2*}(\eta)\big)+\,c_1 D(\eta\|\cP_{\rho_\cM}(\rho))\big)\right)\,, \label{eq_theo_AT_pinching}
	\end{align}
where the infimum in the second line runs over $\eta$ such that $\eta=\cP_{\rho_\cM}(\eta)$ and $\tr_{\cK_i}[P_i\,\eta\,P_i]=\rho_{\cH_i}$.
\end{theorem}

\begin{proof}
	The proof starts similarly to that of Corollary \ref{corollaruweak}. We once again simply need to bound the integral on the right hand side of (\ref{mainequation}). By considering $\eta$ as in the statement of the theorem and writing for the moment $\tilde{d}:=D_{\max}\big(E_{1*}\circ E_{2*}(\rho)\|E_{1*}\circ E_{2*}(\eta)\big)$,  we obtain
	\begin{align*}
	\tr\Big[ \rho_{1}\,\rho_\cM^{\frac{-1-it}{2}}\,\rho_{2}\,\rho_\cM^{\frac{-1+it}{2}}\Big]=		\tr\Big[ \rho\,\rho_\cM^{\frac{-1-it}{2}}\,E_{1*}(\rho_{2})\,\rho_\cM^{\frac{-1+it}{2}}\Big]\le \e^{\tilde{d}}  \,	\tr\Big[ \rho\,\rho_\cM^{\frac{-1-it}{2}}\,E_{1*}\circ E_{2*}(\eta)\,\rho_\cM^{\frac{-1+it}{2}}\Big]\,.
	\end{align*}
	To simplify the notation, let us write: $\eta_{12}:=E_{1*}\circ E_{2*}(\eta)$. Now, note that the following holds:
		\begin{equation*}
	\tr \left[   \left( \rho  - \rho_\cM \right)  \rho_\cM^{\frac{-1-it}{2}}   \left( \eta_{12}   - \rho_\cM \right)  \rho_\cM^{\frac{-1+it}{2}}  \right] 
	= \tr \left[ \rho \, \rho_\cM^{\frac{-1-it}{2}}	\,  \eta_{12} \, \rho_\cM^{\frac{-1+it}{2}} \right] -1 -1 + 1,
	\end{equation*}
	since $ E^\cM_*$, $E_{1*}$ and $E_{2*}$ are conditional expectations in the Schr\"{o}dinger picture and, thus, trace preserving. Therefore,
	\begin{align*}
	& \ln  \int_{-\infty}^{+\infty} \e^{\tilde{d}}  \, \tr \left[ \rho \, \rho_\cM^{\frac{-1-it}{2}} \, \eta_{12} \, \rho_\cM^{\frac{-1+it}{2}} \right] \,\beta_0 (t) \,dt \, \\
	&\phantom{adasdasdad}\phantom{adasdasdad}\phantom{adasdasdad}= \ln \int_{-\infty}^{+\infty} \e^{\tilde{d}}  \, \left( \tr \left[   \left( \rho   - \rho_\cM \right)  \rho_\cM^{\frac{-1-it}{2}}   \left( \eta_{12}   - \rho_\cM \right) \rho_\cM^{\frac{-1+it}{2}}  \right]  + 1  \right)\,\beta_0 (t) \,dt \\
	&\phantom{adasdasdad}\phantom{adasdasdad}\phantom{adasdasdad} \leq  \tilde{d} + \int_{-\infty}^{+\infty}  \tr \left[   \left( \rho   - \rho_\cM \right)  \rho_\cM^{\frac{-1-it}{2}}   \left( \eta_{12}  -\rho_\cM \right) \rho_\cM^{\frac{-1+it}{2}}  \right] \, \beta_0 (t)\,dt\,,
	\end{align*}
	where we have used that $	\ln(x +1)\le x$ for positive real numbers. 	 Defining $X:=\Gamma_{\rho_\cM}^{-1}(\rho)$ and $Y_t:=\rho_\cM^{\frac{-1-it}{2}}\,\eta\, \rho_\cM^{\frac{-1+it}{2}}$, we note that
	\begin{equation}
	E_1 \circ E_2 [Y_t] = \rho_\cM^{\frac{-1-it}{2}}\,\eta_{12}\, \rho_\cM^{\frac{-1+it}{2}},
	\end{equation}
and we can rewrite the previous expression as
	\begin{align}
	& \int_{-\infty}^{+\infty}  \tr \left[   \left(\rho   - \rho_\cM \right) \rho_\cM^{\frac{-1-it}{2}}   \left( \eta_{12}  - \rho_\cM\right) \rho_\cM^{\frac{-1+it}{2}}  \right]  \, \beta_0 (t)\,dt\nonumber \\
	&\phantom{adasdasdad}\phantom{adasdasdad}\phantom{adasdasdad}=  \int_{-\infty}^{+\infty} \tr \left[    \left( X  - E^\cM[X] \right)  \Delta_{\rho_\cM}^{-it/2}   \left( \eta_{12}   - \rho_\cM  \right)  \right] \, \beta_0 (t)\,dt \nonumber\\
	&\phantom{adasdasdad}\phantom{adasdasdad}\phantom{adasdasdad}=  \int_{-\infty}^{+\infty}  \tr \left[    \left(X  - E^\cM[X]  \right)  \rho_\cM^{\frac{1}{2}}    \left( E_1 \circ E_2 [Y_t]   -E^\cM [Y_t] \right)  \rho_\cM^{\frac{1}{2}}  \right] \, \beta_0 (t)\,dt\nonumber  \\
	&\phantom{adasdasdad}\phantom{adasdasdad}\phantom{adasdasdad}= \int_{- \infty}^{+ \infty} \left\langle  X  - E^\cM[X]  , \ E_1 \circ E_2 [Y_t]   -E^\cM [Y_t]  \right\rangle_{\rho_\cM}\, \beta_0 (t)\,dt\, ,
	\end{align}		
thus obtaining the following inequality
	\begin{align*}
	\ln\,	\int_{-\infty}^{\infty}	\e^{\tilde{d}}  \, 	\tr\Big[ \rho_{1}\,\rho_\cM^{\frac{-1-it}{2}}\,\eta_{12}\,\rho_\cM^{\frac{-1+it}{2}}\Big]\,\beta_0(t)\,dt\le \tilde{d}+\int_{-\infty}^\infty \,\langle X-E^\cM[X]\,,E_1\circ E_2[Y_t]-E^\cM[Y_t]\rangle_{\rho_\cM}\,\beta_0(t)\,dt .
	\end{align*}
	
	 Now, we focus on the integrand on the right-hand side of the above inequality. Denote for any $A\in\cB(\cH)$, 
	\[A^{(\lambda,i)}:=\left(\proj{\lambda^{(i)}}\otimes \ind_{\cK_i}\right)P_i\,A\,P_i\left(\proj{\lambda^{(i)}}\otimes \ind_{\cK_i}\right)\,.\]
	We also write $A^{(\lambda,i)}=\proj{\lambda^{(i)}}\otimes A^{(\lambda,i)}$ by a slight abuse of notation. Then
	\begin{align*}
	\langle X-E^\cM[X]\,,&\, E_1\circ E_2[Y_t]-E^\cM[Y_t]\rangle_{\rho_\cM}\\
	&= \sum_{i\in I_\cM}\,\sum_{\lambda^{(i)}\in\operatorname{Spec}(\rho_{\cH_i})}\,\lambda^{(i)}\,\langle X^{(\lambda,i)}-E^\cM[X^{(\lambda,i)}],\,E_1\circ E_2[Y_t^{(\lambda,i)}]-E^\cM[Y_t^{(\lambda,i)}]\rangle_{\tau_i}\,.
	\end{align*}
	
	Next, by Hölder's inequality each summand in the right-hand side above is upper bounded by
	\begin{align*}
	& \|(\id-(E^\cM)^{(i)})[X^{(\lambda,i)}] \|_{\mathbb{L}_1(\tau_i)}\,\|(E_1^{(i)}\circ E_2^{(i)}-(E^\cM)^{(i)}) [Y_t^{(\lambda,i)}] \|_\infty\nonumber\\
	&\phantom{adasdasdad}\phantom{adasdasdad}\phantom{adasdasdad} \le c_1\,  \|(\id-(E^\cM)^{(i)})[X^{(\lambda,i)}] \|_{\mathbb{L}_1(\tau_i)}\,\| (\id-(E^\cM)^{(i)})[Y_t^{(\lambda,i)}]\|_{\mathbb{L}_1(\tau_i)}\nonumber\\
	&\phantom{adasdasdad}\phantom{adasdasdad}\phantom{adasdasdad} = c_1\,\|\rho^{(\lambda,i)}-E^\cM_{*}(\rho^{(\lambda,i)})\|_1\,\|\tau_i'-\tau_i\|_{1}\label{eqchange}\\
	&\phantom{adasdasdad}\phantom{adasdasdad}\phantom{adasdasdad} \le \frac{c_1}2\,\left(\|\rho^{(\lambda,i)}-E^\cM_{*}(\rho^{(\lambda,i)})\|_1^2+\|\tau_i'-\tau_i\|_{1}^2\right)\,,\nonumber
	\end{align*}
where we use Young's inequality in the last line. Using Pinsker's inequality and summing over the indices $i$ and $\lambda^{(i)}$, we find that
	\begin{align*}
	\ln\int_{-\infty}^\infty   \tr\left[ \rho\,\rho_\cM^{\frac{-1-it}{2}}E_{1*}\circ E_{2*}(\cP_\cM(\rho))\,\rho_\cM^{\frac{-1+it}{2}}    \right]  \beta_0(t)\,dt\le \tilde{d} +c_1\,D(\rho\|\rho_\cM)+ c_1\,D(\eta\|\cP_{\rho_\cM}(\rho))\,.
	\end{align*}
    \Cref{eqgeneral} follows after rearranging the term. In order to obtained \Cref{eqchanged}, we exploit that $\rho_\cM$ is a fixed point of $\cP_{\rho_\cM}$ and therefore
    \[D(\rho\|\rho_\cM)=D(\rho\|\cP_{\rho_\cM}(\rho))+D(\cP_{\rho_\cM}(\rho)\|\rho_\cM)\,.\]
    We can then apply \Cref{eqgeneral} to $\cP_{\rho_\cM}(\rho)$ and remark that the weak constant vanishes. The result follows after remarking that $\cP_{\rho_\cM}\circ E_*^\cM=E_*^\cM\circ\cP_{\rho_\cM}$ and applying the data-processing inequality to the map $\cP_{\rho_\cM}$.
\end{proof}

\begin{remark}
	In the case of a classical evolution over a classical system, taking $\eta=\cP_{\rho_\cM}(\rho)$ shows that $d=0$ in \Cref{eq_theo_AT_pinching}, and thus we get back the strong approximate tensorization of \cite{cesi2001quasi}. In \Cref{classicalglauber}, we will see that this remains also true for classical evolution over quantum systems. The estimation of the constant $c$ under a condition of clustering of correlations is discussed in the next section. 
\end{remark}

The next proposition provides a short analysis of the weak constant in \Cref{theo_AT_pinching}. We note that the interpretation of this term as a deviation to the classical case is direct from the pinching argument, which explicitly ``pinches" on a classical basis. However, as opposed to \Cref{HSAT}, we were unable to prove that the weak constant necessarily vanishes when the doubly stochastic conditional expectations form a commuting square.

\begin{proposition}\label{propboundsd1d2}
	With the notations of \Cref{propAT} and \Cref{theo_AT_pinching},
	\begin{equation*}
	 \sup_{\rho\in\cD(\cN)}\,D_{\max}\big(E_{1*}\circ E_{2*}(\rho)\| E_{1*}\circ E_{2*}(\eta)\big)\leq d_1 + d_2
	\end{equation*}
 where
	\begin{align*}
	 &d_1:=\sup_{\rho\in\cD(\cN)}\,D_{\max}\big(E_{1*}\circ E_{2*}(\rho)\| E_{1*}\circ E_{2*}(\cP_\cM(\rho))\big)\,,\\
	&d_2:=\,\max_{i\in I_\cM}\sup_{\rho^{(i)}\in\cD(P_i\cN P_i)}\,I_{\max}\big( \cH_i:\cK_i  \big)_{\rho^{(i)}}\,, 
	\end{align*}
and where $\cP_\cM:=\sum_{i\in I_\cM}P_i(\cdot)P_i$.
    
    Furthermore, given $i\in I_\cN$, denote by $I^{(i)}_\cM$ the set of indices corresponding to the minimal projectors in $\cM$ contained in the $i$-th block of $\cN$. Moreover, for each of the blocks $i$ of $\cN$, of corresponding minimal projector $P^{\cN}_i$, decompose $P^\cN_i\cM P^\cN_i$ as follows: letting $P_i^\cN\cH:= \bigoplus_{j\in I^{(i)}_\cM} \,\cH^{(i)}_{j}\otimes \cK^{(i)}_j$,
	\begin{align*}
	P_i^\cN\cM P_i^\cN:=\bigoplus_{j\in I^{(i)}_\cM} \cB(\cH^{(i)}_j)\otimes \Id_{\cK_j^{(i)}}\,.
	\end{align*}
	Then,
	\begin{align*}
	&d_1\le \max_{i\in I_\cN}\ln(|I^{(i)}_\cM|)\,\\
	&d_2\le   2\max_{i\in I_\cN}\max_{j\in I^{(i)}_\cM}\min\left\lbrace \ln (d_{\cH^{(i)}_j}),\ln(d_{\cK^{(i)}_j}) \right\rbrace \,.
	\end{align*}
\end{proposition}

The proof of this result is deferred to Appendix \ref{subsec:proof_estimation_c,d}.

 \subsection{Clustering of correlations}\label{sec:clustering}

In this section we shift slightly our focus and study the multiplicative constant of the previous results, instead of the additive one. More specifically, we provide an interpretation of the multiplicative constant appearing in the last section in terms of certain notions of clustering of correlations. The latter play a particularly relevant role when applied in the context of quantum spin lattices \cite{capel2020modified}.  

The constant $c_1:=\max_{i\in I_\cM}\|E_{1}^{(i)}\circ E_{2}^{(i)}-(E^\cM)^{(i)}:\,\mathbb{L}_1(\tau_i)\to\mathbb{L}_\infty\|$ appearing in \Cref{theo_AT_pinching}
provides a bound on the following covariance-type quantity: For any $i\in I_\cM$ and any $X, Y\in \mathbb{L}_1(\tau_i)$, 
\begin{align}
\operatorname{Cov}_{\CC\Id_{\cK_i},\tau_i}(E_1^{(i)}[X],E_2^{(i)}[Y])&:=\langle  E_1^{(i)}[X]-(E^\cM)^{(i)}[X],\,E_2^{(i)}[Y]-(E^\cM)^{(i)}[Y]\rangle_{\tau_i}\nonumber\\
&\le \,c_1\,\|X\|_{\mathbb{L}_1(\tau_i)}\,\|Y\|_{\mathbb{L}_1(\tau_i)}\,.\label{condL1sclust}
\end{align}
We call the above property \textit{conditional $\mathbb{L}_1$ clustering of correlations}, and denote it by $\operatorname{cond\mathbb{L}}_1(c_1)$.  Conversely, one can show by duality of $\mathbb{L}_p$-norms that if $\operatorname{cond\mathbb{L}}_1(c_1')$ holds for some positive constant $c_1'$, then 
$c_1\le c_1'$: for all $i\in I_\cM$
\begin{align*}
\|E_{1}^{(i)}\circ E_{2}^{(i)}-(E^\cM)^{(i)}:\,\mathbb{L}_1(\tau_i)\to\mathbb{L}_\infty\|&=\sup_{X\in \mathbb{L}_1(\tau_i)}\,\|E_{1}^{(i)}\circ E_{2}^{(i)}[X]-(E^\cM)^{(i)}[X]\|_\infty\\
&=\sup_{X,Y\in \mathbb{L}_1(\tau_i)}\,\langle Y,\,E_1^{(i)}\circ E_2^{(i)}[X]-(E^\cM)^{(i)}[X]\rangle_{\tau_i}\\
&=\sup_{X,Y\in \mathbb{L}_1(\tau_i)}\,\langle E_1^{(i)}[Y]-(E^\cM)^{(i)}[Y],\, E_2^{(i)}[X]-(E^\cM)^{(i)}[X]\rangle_{\tau_i}\\
&\le c_1'\,.
\end{align*}

In \cite{Kastoryano2014}, the authors introduced a different notion of clustering of correlation in order to show the positivity of the spectral gap of Gibbs samplers\footnote{Here, we formulate it in our general framework of finite-dimensional $*$-algebras.}.
\begin{definition}\label{defL2strong}
	We say that $\cM\subset \cN_1,\cN_2\subset \cN$ satisfies \textit{strong $\mathbb{L}_2$ clustering of correlations} with respect to the state $\sigma\in\cD(\cM)$ with constant $c_{2}>0$ if for all $X,Y\in\cN$,
	\begin{align}\label{L2_clust}
	\operatorname{Cov}_{\cM,\sigma}(E_1[X],E_2[Y])\le \,c_{2}\,    \|X\|_{\mathbb{L}_2(\sigma)}\,\|Y\|_{\mathbb{L}_2(\sigma)}\,.
	\end{align}
	Equivalently, $\|E_1\circ E_2-E^\cM:\,\mathbb{L}_2(\sigma)\to \mathbb{L}_2(\sigma)\|\le c_2$.
\end{definition}

 \Cref{defL2strong} does not depend on the state $\sigma\in \cD(\cM)$ chosen. 
This is the content of the next theorem, whose proof is presented in Appendix \ref{subsec:proof_clustering}.

\begin{lemma}\label{theo_L2clust}
	Let $\cM\subset \cN_1,\cN_2\subset\cN\subset \cB(\cH)$ be von Neumann subalgebras of the algebra $\cB(\cH)$ so that $\cN_1 \cap \cN_2 \neq \emptyset$. Then, for any two states $\sigma,\sigma'\in\cD(\cM)$:
	\begin{equation*}
	\sup_{X\in\cN}\,\frac{\operatorname{Cov}_{\cM,\,\sigma}(E_1[X],E_2[X])}{\|X\|_{\mathbb{L}_2(\sigma)}^2} = \sup_{X\in\cN}\,\frac{\operatorname{Cov}_{\cM,\,\sigma'}(E_1[X],E_2[X])}{\|X\|_{\mathbb{L}_2(\sigma')}^2}
	\end{equation*}
\end{lemma}

\begin{remark}
As a consequence of the previous theorem, we realize that the condition assumed in \cite{Kastoryano2014} of strong $\mathbb{L}_2$ clustering of correlation with respect to one invariant state, to prove positivity of the spectral gap for the Davies dynamics, would be analogous to assuming strong $\mathbb{L}_2$ clustering of correlation with respect to any invariant state. 
\end{remark}

It is easy to see that the above notion of strong $\mathbb{L}_2$ clustering of correlation implies that of a \textit{conditional $\mathbb{L}_2$ clustering}, denoted by $\operatorname{cond\mathbb{L}_2}(c_2)$, simply defined by replacing the $\mathbb{L}_1$ norms by $\mathbb{L}_2$ norms in \Cref{condL1sclust}, or equivalently by assuming that 
\[
\max_{i\in I_\cM}\|E_1^{(i)}\circ E_2^{(i)}-(E^\cM)^{(i)}:\,\mathbb{L}_2(\tau_i)\to \mathbb{L}_2(\tau_i)\|\le c_2\,.      \]

One can ask whether the converse holds. We prove it under the technical assumption that the composition of conditional expectations $E_1\circ E_2$ cancels off-diagonal terms in the decomposition of $\cM$:
\begin{equation}\label{eq_d1=0}
E_{1*}\circ E_{2*}=E_{1*}\circ E_{2*}\circ\cP_\cM\,. 
\end{equation}
This is for instance the case when $\cM\subset\cN_1,\cN_2\subset\cN$ forms a commuting square.

\begin{proposition}\label{theo_transference}
Assume that \Cref{eq_d1=0} holds. Then:
 \begin{enumerate}
  \item $d_1=0$ in \Cref{propboundsd1d2} and
  \item strong $\LL_2$ clustering is equivalent to conditional $\LL_2$ clustering.
 \end{enumerate}
\end{proposition}
The proof for this result is also deferred to Appendix \ref{subsec:proof_clustering}.

We conclude this section by noting a crutial difference between $\mathbb{L}_2$ and $\mathbb{L}_1$ clusterings: similarly to \Cref{defL2strong}, one could define a notion of \textit{strong $\mathbb{L}_1$ clustering of correlation} with respect to a state $\sigma\in\cD(\cM)$:
\begin{align*}
\|E_1\circ E_2-E^\cM:\,\mathbb{L}_1(\sigma)\to \mathbb{L}_\infty(\cN)\|\equiv c_1(\sigma)\,<\infty.
\end{align*}
This would in particular imply $\operatorname{cond\mathbb{L}_1}(c_1(\sigma))$. With this notion, and from an argument very similar to that of the proof of \Cref{theo_AT_pinching}, we could show the following bound on the error term in \Cref{propAT}:
	\begin{align*}
	\ln \left\lbrace \int_{-\infty}^\infty\,\tr\left[\rho_{1}\rho_{\cM}^{\frac{-1-it}{2}}\rho_{2}\rho_{\cM}^{\frac{-1+it}{2}}   \right]\,\beta_0(t)\,dt \right\rbrace&\le  	 2\,\|E_{1}\circ E_2-E^{\cM}:\,\mathbb{L}_1(\rho_\cM)\to\mathbb{L}_\infty\|\,D(\rho\|\rho_\cM)\,\\
	&\equiv 2 \,c_1(\rho_\cM)\,D(\rho\|\rho_\cM)\,.
	\end{align*}

From this, one would conclude a strong approximate tensorization result if one could find a uniform bound on $c_1(\sigma)$ for any $\sigma\in\cD(\cM)$. However, and as opposed to the case of strong $\mathbb{L}_2$ clustering, the constant $c_1(\sigma)$ depends on the state $\sigma$, and can in particular diverge: this is the case whenever there exists $i\in I_\cM$ such that $\dim(\cH_i)<\infty$, and for a state $\sigma:=|\psi\rangle\langle\psi|_{\cH_i}\otimes \tau_i$ that is pure on $\cH_i$. This justifies our choice of $\operatorname{cond\mathbb{L}}_1$ as the better notion of $\mathbb{L}_1$ clustering in the quantum setting. After the submission of this manuscript, new insights into this particular problem were shed in \cite{gao2021spectral}. We defer a discussion of their results to \Cref{sec:conclusion}.

\section{Applications}\label{sec:applications}

This section is devoted to two applications of the results of last section. In \Cref{subsec:application_MLSI}, we show the usefulness of Theorem \ref{theo_AT_pinching} in the context of modified logarithmic Sobolev inequalities. Then, we derive new entropic uncertainty relations in \Cref{subsec:uncertainty_relations}.

\subsection{Modified logarithmic Sobolev inequalities for biased bases}\label{subsec:application_MLSI}

Take $\cH=\C^l$ and assume that the algebra $\cN_1$ is the diagonal onto some orthonormal basis $|e^{(1)}_k\rangle$, whereas $\cN_2$ is the diagonal onto the basis $|e^{(2)}_k\rangle$. Moreover, choose $\cM$ to be the trivial algebra $\CC\Id_\ell$. Hence for each $i\in\{1,2\}$, $E_i$ denotes the Pinching map onto the diagonal $\operatorname{span}\left\lbrace |e_k^{(i)}\rangle\langle e_k^{(i)}|\right\rbrace$ and $E^\cM=\frac{\Id}{\ell}\tr[\cdot]$. Then, for any $X\ge0$:	 
	\begin{align*}
	\|(E_1\circ E_2-E^\cM)(X)\|_\infty&=\Big\|\sum_{k,k'}|e^{(1)}_k\rangle\langle e^{(1)}_k|e^{(2)}_{k'}\rangle\langle e^{(2)}_{k'}| X| e^{(2)}_{k'}\rangle\langle e^{(2)}_{k'}|e^{(1)}_k\rangle\langle e^{(1)}_k|-\frac{1}{\ell}\,|e^{(1)}_k\rangle\langle e^{(1)}_k|\,\langle e^{(2)}_{k'}|X|e^{(2)}_{k'}\rangle\Big\|_\infty\\
	&=\Big\|\sum_{k,k'}\,\big( |\langle e^{(1)}_k|e^{(2)}_{k'}\rangle |^2-\frac{1}{\ell} \big)|e^{(1)}_k\rangle\langle e^{(1)}_k|\,\langle e^{(2)}_{k'}| X| e^{(2)}_{k'}\rangle\Big\|_\infty\\
	&=\max_{k}\,\sum_{k'}\,\Big| |\langle e^{(1)}_k|e^{(2)}_{k'}\rangle |^2-\frac{1}{\ell} \Big| \quad\langle  e^{(2)}_{k'}| X| e^{(2)}_{k'}\rangle\\
	&\le \eps\,\frac{1}{\ell}\tr[X]
	\end{align*}
	where $\eps:=\ell \,\max_{k,k'}\,\Big| |\langle e^{(1)}_k|e^{(2)}_{k'}\rangle |^2-\frac{1}{\ell} \Big|$. Hence
	\begin{align*}
	\|(E_1\circ E_2-E^\cM):\mathbb{L}_1(\ell^{-1}\Id)\to \mathbb{L}_\infty\|\le \eps \,,  
	\end{align*}
	so that by choosing $\eta=\rho=\cP_{\rho_\cM}(\rho)$ in Theorem \ref{theo_AT_pinching}, as long as $\eps<1$, for any $\rho\in \cD(\CC^\ell)$, AT($(1-\eps)^{-1},0$) holds:
	\begin{align}\label{our-bound}
	D(\rho\| \ell^{-1}\Id)\le \frac{1}{1-\eps}\,(D(\rho\|E_{1*}(\rho))+D(\rho\|E_{2*}(\rho)))\,.
	\end{align}

This result is related to Example 4.5 of \cite{nick2019scoopingpaper}. There, the author obtains an inequality that can be rewritten in the following form:
	\begin{align}\label{nick-bound}
	D(\rho\| \ell^{-1}\Id)\le 4 \left(  \ln_{1-\delta} \left( \frac{2}{3 \ell + 5}   \right) + 1   \right)   \,(D(\rho\|E_{1*}(\rho))+D(\rho\|E_{2*}(\rho)))\,,
	\end{align}
where $\delta$ here is related with $\varepsilon$ in our example by:
\begin{equation*}
\delta \geq  \frac{1}{\ell} (1- \varepsilon).
\end{equation*}


The approximate tensorization derived in (\ref{our-bound}) can be used to find exponential convergence in relative entropy  of the primitive quantum Markov semigroup $\e^{t\cL}$, where 
\begin{align*}
\cL(X):=E_1(X)+E_2(X)-2X\,.
\end{align*}
Indeed, for any state $\rho\in\cD(\cH)$, denoting by $\rho_t$ the evolved state $\e^{t\cL}(\rho)$ up to time $t$, the fact that $D(\rho_t\|\ell^{-1}\Id)\le \e^{-\alpha t}D(\rho\|\ell^{-1}\Id)$ holds for some $\alpha >0$ is equivalent to the so-called modified logarithmic Sobolev inequality. Let us recall that $\cL$ is said to satisfy a positive \textit{modified logarithmic Sobolev inequality} (MLSI for short) if there exits a constant $\alpha >0$ such that the following inequality holds for every $\rho \in \mathcal{D}(\mathcal{H})$:
\begin{equation*}
     \alpha D(\rho\|\ell^{-1}\Id) \leq - \tr[ \cL(\rho) ( \ln \rho - \ln \ell^{-1}\Id)] \, .
\end{equation*}
In such a case, the optimal $\alpha$ for which the previous inequality holds is called the \textit{modified logarithmic Sobolev constant}. In this particular setting, by \cite[Lemma 3.4]{junge2019stability}, the MLSI for $\cL$ can be written as
\begin{align*}
  \alpha D(\rho\|\ell^{-1}\Id)\le  D(\rho\|E_{1*}(\rho))+D(E_{1*}(\rho)\|\rho)+D(\rho\|E_{2*}(\rho))+D(E_{2*}(\rho)\|\rho)\,.
\end{align*}
By positivity of the relative entropy, it suffices to prove the existence of a constant $\alpha>0$ such that 
\begin{align*}
\alpha D(\rho\|\ell^{-1}\Id)\le D(\rho\|E_{1*}(\rho))+D(\rho\|E_{2*}(\rho))\,.
\end{align*}
This last inequality is equivalent to (\ref{our-bound}) for $\alpha=1-\eps$. Therefore, Theorem \ref{theo_AT_pinching} yields as a consequence the fact that the generator $\cL$ defined above satisfies a MLSI of constant bounded by $1-\varepsilon$. 

\subsection{Tightened entropic uncertainty relations}\label{subsec:uncertainty_relations}

Given a function $f\in\mathbb{L}_2(\mathbb{R})$ and its Fourier transform $\cF[f]$ with $\|f\|_{\mathbb{L}_2(\mathbb{R})}=\|\cF[f]\|_{\mathbb{L}_2(\mathbb{R})}=1$, Weyl proved in \cite{weyl1950theory} the following uncertainty relation:
\begin{align}\label{uncertainty}
V(|f|^2)\,V(|\cF[f]|^2)\ge \frac{1}{16\pi^2}\,,
\end{align}
where, given a probability distribution function $g$, $V(g)$ denotes its variance. The uncertainty inequality means that $|f|^2$ and $|\cF[f]|^2$ cannot both be concentrated arbitrarily close to their corresponding means. An entropic strenghthening of (\ref{uncertainty}) was derived independently by Hirschmann \cite{hirschman1957note} and Stam \cite{stam1959some}, and tightened later on by Beckner \cite{beckner1975inequalities}:
\begin{align*}
H(|f|^2)+H(|\cF[f]|^2)\ge \ln\frac{\e}{2}\,,
\end{align*}
where $H(g):=-\int_{\mathbb{R}} \,g(x)\ln g(x)\,dx$ stands for the differential entropy functional. In the quantum mechanical setting, this inequality has the interpretation that the total amount of uncertainty, as quantified by the entropy, of non-commuting observables (i.e. the position and momentum of a particle) is uniformly lower bounded by a positive constant independently of the state of the system. For an extensive review of entropic uncertainty relations for classical and quantum systems, we refer to the recent survey \cite{coles2017entropic}. 

More generally, given two POVMs $\mathbf{X}:= \{X_x\}_{x}$ and $\mathbf{Y}:=\{Y_y\}_{y}$ on a quantum system $A$, and in the presence of side information $M$ that might help to better predict the outcomes of $\mathbf{X}$ and $\mathbf{Y}$, the following state-dependent tightened bound was found in \cite{frank2013extended} (see also \cite{berta2010uncertainty} for the special case of measurements in two orthonormal bases and \cite{maasen1988uncertainty} for the case without memory): for any bipartite state $\rho_{AM}\in\cD(\cH_A\otimes \cH_M)$,
\begin{align}\label{uncert}
S(X|M)_{(\Phi_{\mathbf{X}}\otimes \id_M)(\rho)}+S(Y|M)_{(\Phi_{\mathbf{Y}}\otimes \id_M)(\rho)}\ge- \ln c'+S(A|M)_\rho\,,
\end{align}
with $c'=\max_{x,y}\{\tr(X_x\,Y_x)\}$, where $\Phi_{\mathbf{Z}}$ denotes the quantum-classical channel corresponding to the measurement $\mathbf{Z}\in\{\mathbf{X},\mathbf{Y}\}$:
\begin{align*}
\Phi_{\mathbf{Z}}(\rho_A):=\sum_{z}\,\tr(\rho_A Z_z)\,|z\rangle\langle z|_Z\,.
\end{align*}

The above inequality has been recently extended to the setting where the POVMs are replaced by two arbitrary quantum channels in \cite{gao2018uncertainty}. In this section, we restrict ourselves to the setting of \cite{berta2010uncertainty}, so that the measurement channels reduce to the Pinching maps of \Cref{subsec:application_MLSI}. First of all, we notice that the relation (\ref{uncert}) easily follows from \Cref{corollaruweak}:

\begin{example}\label{ex:example3}
Take $\cH_{AM}=\cH_A \otimes \cH_M$ a bipartite system and, as in the case of \Cref{subsec:application_MLSI}, assume that the algebra $\cN_1$ is the diagonal onto some orthonormal basis $|e^{(\mathcal{X})}_x\rangle$ in $\cH_A $, whereas $\cN_2$ is the diagonal onto the basis $|e^{(\mathcal{Y})}_y\rangle$ also in $\cH_A $. Moreover, choose $\cM$ to be the algebra $\CC\Id_\ell \otimes \mathcal{B}(\cH_M)$. Hence for each alphabet $\mathcal{Z}\in\{\mathcal{X},\mathcal{Y}\}$, $E_\mathcal{Z}$ denotes the Pinching map onto the diagonal $\operatorname{span}\left\lbrace |e_z^{(\mathcal{Z})}\rangle\langle e_z^{(\mathcal{Z})}|\right\rbrace$, which we tensorize with the identity map in $M$, and $E^\cM \otimes \id_M=\frac{1}{d_A}\Id_A \otimes\tr_A[\cdot] $. Then, for every $\rho_{AM} \in \mathcal{D}(\cH_{AM})$,
\begin{align*}
S(X|M)_{(E_{\mathcal{X}}\otimes \id_M)(\rho_{AM})} & = - D\left( (E_{\mathcal{X}}\otimes \id_M)(\rho_{AM}) \Big| \Big| \frac{\Id_A}{d_A} \otimes \rho_M \right) +\ln d_A \\
& = D(\rho_{AM} || (E_{\mathcal{X}}\otimes \id_M)(\rho_{AM})) - D(\rho_{AM} || (E^{\mathcal{M}}\otimes \id_M)(\rho_{AM})) +\ln d_A,
\end{align*}
where the last equality is derived from \cite[Lemma 3.4]{junge2019stability}. Hence, since
\begin{equation*}
D(\rho_{AM} || (E^{\mathcal{M}}\otimes \id_M)(\rho_{AM})) = - S(A| M)_{\rho_{AM}} + \ln d_A ,
\end{equation*}
by virtue of Corollary \ref{corollaruweak} we have
\begin{align*}
S(X|M)_{(E_{\mathcal{X}}\otimes \id_M)(\rho_{AM})} &+ S(Y|M)_{(E_{\mathcal{Y}}\otimes \id_M)(\rho_{AM})} \\
& \geq  D(\rho_{AM} || (E_{\mathcal{X}}\otimes \id_M)(\rho_{AM}))+  D(\rho_{AM} || (E_{\mathcal{Y}}\otimes \id_M)(\rho_{AM}))\\
& \; \; \; \;  - 2 D(\rho_{AM} || (E^{\mathcal{M}}\otimes \id_M)(\rho_{AM})) + 2 \ln d_A\\
& \geq S(A|M)_{\rho_{AM}} - d + \ln d_A ,
\end{align*}
where 
\begin{equation*}
d:=\sup_{\rho\in\cD(\cN_{2})}\inf\big\{\ln(\lambda)|\,E_{\mathcal{X}} \otimes \id_M (\rho)\le \lambda\eta\,\text{ for some }\,\eta\in\cD(\cM)  \big\} .
\end{equation*}

Now, taking into account the computations of \Cref{subsec:application_MLSI}, notice that 
\begin{equation*}
d= \ln \left(  d_A \, \underset{x, \, y}{\text{max}}\, |\langle e^{(\mathcal{X})}_x|e^{(\mathcal{Y})}_{y}\rangle |^2  \right) \, , 
\end{equation*}
obtaining thus expression (\ref{uncert}).
\end{example}

 However, close to the completely mixed state, this inequality is not tight  whenever $\mathbf{X}$ and $\mathbf{Y}$ are not mutually unbiased bases (i.e. $\exists x\in\cX,y\in\mathcal{Y}$ such that $|\langle X^x|Y^y\rangle|^2>\frac{1}{d_A}$). Here, we derive the following strengthening of \Cref{uncert} when $d_M=1$ as a direct consequence of  \Cref{theo_AT_pinching}:
\begin{corollary}\label{cor:2pinching}
	Given a finite alphabet $\mathcal{Z}\in \{\mathcal{X},\mathcal{Y}\}$, let $E_{\mathcal{Z}}$ denote the Pinching channels onto the orthonormal basis $\{|e^{(\mathcal{Z})}_z\rangle\}_{z\in\mathcal{Z}}$ corresponding to the measurement $\mathbf{Z}$. Assume further that $c_1= d_A\max_{x,y}\big| |\langle e^{(\mathcal{X})}_x|e^{(\mathcal{Y})}_y\rangle|^2-\frac{1}{d_A}   \big|<1$.	Then the following strenghtened entropic uncertainty relation holds for any state $\rho\in\cD(\cH_{A})$, 
	\begin{align}
	S(X)_{E_{\mathcal{X}}(\rho)}+S(Y)_{E_{\mathcal{Y}}(\rho)}\ge  (1+c_1)\,S(A)_\rho+(1-c_1)\ln d_A\,.
	\end{align}
\end{corollary}

\begin{proof}
Following the first lines of Example \ref{ex:example3} for $d_M=1$, we have
\begin{equation*}
S(X)_{E_{\mathcal{X}}(\rho)}+  S(Y)_{E_{\mathcal{Y}}(\rho)} = D(\rho || E_{\mathcal{X}}(\rho)) + D(\rho || E_{\mathcal{Y}}(\rho)) - 2 D(\rho || E^{\mathcal{M}}(\rho)) + 2 \ln d_A ,
\end{equation*}
where $E^\cM=\frac{\Id}{\ell}\tr[\cdot]$. Then, by virtue of Theorem \ref{theo_AT_pinching}, 
\begin{align*}
 D(\rho || E_{\mathcal{X}}(\rho)) &+ D(\rho || E_{\mathcal{Y}}(\rho)) - 2 D(\rho || E^{\mathcal{M}}(\rho)) \\
 & \geq (-1 - c_1) D(\rho || E^{\mathcal{M}}(\rho)) - D_\text{max} (E_\mathcal{X} \circ E_\mathcal{Y} (\rho) \|  E_\mathcal{X} \circ E_\mathcal{Y} (\eta) ) - c_1 D (\eta \| \mathcal{P}_{\rho_\mathcal{M}}(\rho) ) \, 
\end{align*}
for any $\eta = \mathcal{P}_{\rho_\mathcal{M}}(\rho)$, and by further choosing $\eta=\rho$, the last two terms above vanish. Thus, we have:
\begin{equation*}
S(X)_{E_{\mathcal{X}}(\rho)}+  S(Y)_{E_{\mathcal{Y}}(\rho)} \geq  (-1 - c_1) \, D(\rho || E^{\mathcal{M}}(\rho)) + 2 \ln d_A .
\end{equation*}
To conclude, just notice that 
\begin{equation*}
D(\rho || E^{\mathcal{M}}(\rho)) = - S(A)_{\rho} + \ln d_A ,
\end{equation*}

\end{proof}


Analogously, we can study the case for three different orthonormal bases (see \cite{berta2010uncertainty}). For that, let us recall that given $ \cN_1, \cN_2, \cN_3 \subset \cN$ von Neumann subalgebras and $\cM \subset \cN_1 \cap \cN_2 \cap \cN_3$, if we consider their associated conditional expectations $E_i$ with respect to a state $\sigma$, and for each pair $(\cN_i, \cN_j)$ a result of AT($c_{ij}, d_{ij}$) holds, then for every $\rho \in \cD(\cN)$:
\begin{equation}
D(\rho || E_*^\cM (\rho)) \leq \frac{2}{3} \, \underset{i, \, j \in \qty{1,2,3}}{\text{max}} \qty{c_{ij} } \left( \, D(\rho || E_{1*} (\rho) ) + D(\rho || E_{2*} (\rho) ) + D(\rho || E_{3*} (\rho) ) \, \right) + \frac{d_{12} + d_{13} + d_{23}}{3}.
\end{equation}

\begin{corollary}\label{cor:3pinching}
	Given a finite alphabet $I \in \{\mathcal{X},\mathcal{Y}, \mathcal{Z}\}$, and using the same notation that in Corollary \ref{cor:2pinching}, assume that
	 $$c_1= d_A \underset{I, \, J \in \qty{\mathcal{X}, \mathcal{Y}, \mathcal{Z}}}{\text{max}} \left| |\langle e^{(I)}_i|e^{(J)}_j\rangle|^2-\frac{1}{d_A}   \right|<1 .$$	
	 Then the following strenghtened entropic uncertainty relation holds for any state $\rho\in\cD(\cH_{A})$, 
	\begin{align*}
	S(X)_{E_{\mathcal{X}}(\rho)}+S(Y)_{E_{\mathcal{Y}}(\rho)}+S(Z)_{E_{\mathcal{Z}}(\rho)}\ge  (2+c_1)\,S(A)_\rho+(1-c_1)\ln d_A\,.
	\end{align*}
\end{corollary}

\section{Lattice spin systems with commuting Hamiltonians}\label{sec:dabies}

 In this section, we further control the strong and weak constants appearing in \Cref{theo_AT_pinching} in the context of lattice spin systems, and compare them with previous conditions in the classical and quantum literature. The main result presented in this section is \Cref{prop_classical_inftyT}, where we show that the classical Glauber dynamics embedded in a quantum system satisfies a strong approximate tensorization AT$(1,0)$ at infinite temperature and presents an approximate tensorization AT$(c,0)$ with small multiplicate constant when the temperature is high enough.  This result is contained in  Section \ref{classicalglauber}. 
 
 


 Given a finite lattice $\Lambda \subset\subset \ZZ^d$, we define the tensor product Hilbert space $\cH:=\cH_\Lambda\equiv\bigotimes_{k\in\Lambda}\cH_k$, where for each $k\in\Lambda$, $\cH_k\simeq \CC^\ell$, $\ell\in\NN$. Then, let $\Phi:\Lambda\to \cB_{\operatorname{sa}}(\cH_{\Lambda}) $ be an $r$-local potential, i.e. for any $j\in \Lambda$, $\Phi(j)$ is self-adjoint and supported on a ball of radius $r$ around site $j$. We assume further that $\| \Phi(j) \|\le K$ for some constant $K<\infty$. The potential $\Phi$ is said to be a \textit{commuting potential} if for any $i,j\in\Lambda$, $[\Phi(i),\Phi(j)]=0$.  Given such a local, commuting potential, the Hamiltonian on a subregion $A\subseteq \Lambda$ is defined as
\begin{align}
H_A=\sum_{j\in A}\,\Phi(j)\,.
\end{align}
Next, the corresponding Gibbs state corresponding to the region $A$ and at inverse temperature $\beta$ is defined as
\begin{align}
\sigma_A=\frac{\e^{-\beta H_A}}{\tr[\e^{-\beta H_A}]}\,.
\end{align}
Note that this is in general not equal to the state $\tr_B[\sigma_\Lambda]$.

We begin by introducing Davies semigroups on lattice spin systems. These are the most studied examples of Markovian dynamics studied in this context, together with heat-bath generators defined through Petz recovery maps \cite{BardetCapelLuciaPerezGarciaRouze-HeatBath1DMLSI-2019,[BK16],Temme2014}. Thanks to \Cref{theo_equal_cond}, we know that the conditional expectations arising from both dynamics coincide. Hence, for the rest of the paper, all the results presented will be independent of the choice of underlying dynamics.  

\subsection{Davies generators on lattice spin systems}

Consider the setting introduced in Section \ref{subsec:conditional_expectations} and, in particular, the Hamiltonian modelling the system-bath interaction. As mentioned before, the evolution on the system can be approximated by a quantum Markov semigroup whose generator is of the following form:
\begin{align}\label{totalgenerator}
\cL^{\operatorname{D},\beta}_\Lambda(X)=i[H_\Lambda,X]+\sum_{k\in\Lambda}\,\cL^{\operatorname{D},\beta}_k(X)\,,
\end{align}
where 
\begin{align}\label{lindblad}
\cL^{\operatorname{D},\beta}_k(X)=\sum_{\omega,\alpha}\,\chi^{\beta}_{\alpha,k}(\omega)\,\Big(  S_{\alpha,k}^*(\omega)XS_{\alpha,k}(\omega)-\frac{1}{2}\,\big\{  S_{\alpha,k}^*(\omega)S_{\alpha,k}(\omega),X \big\}   \Big)\,.
\end{align}
Similarly, define the generator $\cL^\beta_A$ by restricting the sum in \Cref{totalgenerator} to the sublattice $A$:
\begin{align}\label{localgenerator}
\cL^{\operatorname{D},\beta}_A(X)=i[H_A,X]+\sum_{k\in A}\,\cL^{\operatorname{D},\beta}_k(X)\,.
\end{align}
Note that $\cL^{\operatorname{D},\beta}_A$ acts non-trivially on the boundary of $A$, denoted by $A_\partial:=\{k\in\Lambda:\,d(k,A)\le r\}$. Then, for any region $A\subset \Lambda$, we define the conditional expectation onto the algebra $\cN_A$ of fixed points of $\cL_A$ with respect to the Gibbs state $\sigma=\sigma_\Lambda$ as follows \cite{Kastoryano2014}: given an adequate decomposition $\cH_\Lambda:=\bigoplus_{i\in I_{\cN_A}}\,\cH_i^A\otimes \cK_i^A$ of the total Hilbert space $\cH_\Lambda$ of the lattice spin system,
\begin{align}\label{Daviescond}
& \cN_A=\bigoplus_{i\in I_{\cN_A}}\,\cB(\cH_i^A)\otimes \Id_{\cK_i^A}\,, \nonumber \\
& E^{\operatorname{D},\beta}_A[X]:=\lim_{t\to \infty}\e^{t\cL^{\operatorname{D},\beta}_A}(X)\equiv \sum_{i\in I_{\cN_A}}\, \tr_{\cK^A_i}(P^{A}_i (\Id_{\cH_i^A}\otimes \sigma_i^A)\,X  P^{A}_i) \otimes \Id_{\cK^A_i}\,,
\end{align}
for some fixed full-rank states $\sigma_i^A$ on $\cK_i^A$. It was shown in Lemma 11 of \cite{Kastoryano2014} that the generator of the Davies semigroups corresponding to a local commuting potential is \textit{frustration-free}. This means that the state $\sigma$ is in the kernels of all $\cL_A^{\operatorname{D},\beta}$, $A\subseteq \Lambda$. Therefore, the conditional expectations $E^{\operatorname{D},\beta}_A$ are all defined with respect to $\sigma$.

In the next section, we study the weak approximate tensorization of the conditional expectations $E_A^{D,\beta}\equiv E_A^\beta$ in the case of a classical Hamiltonian. We start with the following simple observation for commuting Hamiltonians.

 \begin{proposition}\label{prop_cond_ABdisjoint}
  Let $A,B\subset\Lambda$ be two regions separated by at least a distance $2r$, that is such that $A_\partial\cap B_\partial=\emptyset$. Then $\cN_A$ and $\cN_B$ form a commuting square, that is,
  \begin{equation}\label{eq_prop_cond_ABdisjoint1}
   E_A^{\beta}\circ E_B^{\beta}=E_B^{\beta}\circ E_A^{\beta} = E_{A\cup B}^{\beta}\,.
  \end{equation}
  Consequently, for all $\rho\in\cD(\cH_\Lambda)$,
  \begin{equation}\label{eq_prop_cond_ABdisjoint2}
   D\big(\rho\|E_{A\cup B*}^{\beta}(\rho)\big)\leq D\big(\rho\|E_{A*}^{\beta}(\rho)\big)+D\big(\rho\|E_{B*}^{\beta}(\rho)\big)\,.
  \end{equation}
 \end{proposition}

 \begin{proof}
 Remark that by definition of the map $\cL_A^{D,\beta}$, it only acts non-trivially on $A_\partial$ and as identity on $(A_\partial)^c$. Consequently, as $E_A=\lim_{t\to\infty} e^{t\cL_A^{D,\beta}}$, this property carries over to the conditional expectation and we have $E_A=E_A\otimes \ind_{\cH_{A_\partial^c}}$ by slight abuse of notations. Similarly, $E_B=E_B\otimes \ind_{\cH_{B_\partial^c}}$. This shows the result since $A_\partial\cap B_\partial=\emptyset$.
 \end{proof}


\subsection{Classical Hamiltonian over quantum systems}\label{classicalglauber}

In this section, we investigate the case of a quantum lattice spin system undergoing a classical Glauber dynamics, whose framework was already studied in \cite{Cubitt2015}. These semigroups correspond to Davies generators whose Hamiltonian is classical, that is, diagonal in a product basis of $\cH_\Lambda$. In order to make the connection with the classical Glauber dynamics over a classical system (i.e. initially diagonal in the product basis), we introduce the generator more explicitly: consider a lattice spin system over $\Gamma=\ZZ^d$ with classical configuration space $S=\{+1,-1\}$, and, for each $\Lambda\subset \Gamma$, denote by $\Omega_\Lambda=S^\Lambda$ the space of configurations over $\Lambda$. Next, given a classical finite-range, translationally invariant potential $\{J_A\}_{A\in\Gamma}$ and a boundary condition $\tau\in\Omega_{\Lambda^c}$, define the Hamiltonian over $\Lambda$ as

\begin{align*}
H_\Lambda^\tau(\sigma)=-\sum_{A\cap\Lambda\ne 0}\,J_A(\sigma\times\tau),~~~~~~\forall\sigma\in\Omega_\Lambda\,.
\end{align*}
The classical Gibbs state corresponding to such Hamiltonian is then given by 
\begin{align*}
\mu_\Lambda^\tau(\sigma)=(Z_{\Lambda}^\tau)^{-1}\,\exp\big( -H_{\Lambda}^\tau(\sigma)\big)\,,
\end{align*}
Next, define the Glauber dynamics for a potential $J$ as the Markov process on $\Omega_\Lambda$ with the generator 
\begin{align*}
(L_\Lambda f)(\sigma)=\sum_{x\in\Lambda}\,c_{J}(x,\sigma)\nabla_xf(\sigma)\,,
\end{align*}
where $\nabla_xf(\sigma)=f(\sigma^x)-f(\sigma)$ and $\sigma^x$ is the configuration obtained by flipping the spin at position $x$. The numbers $c_J(x,\sigma)$ are called transition rates and must satisfy the following assumptions: 
\begin{itemize}
	\item[1.] There exist $c_m,c_M$ such that $0<c_m\le c_J(x,\sigma)\le c_M<\infty$ for all $x,\sigma$. 
	\item[2.] $c_J(x,.)$ depends only on spin values in $b_r(x)$.
	\item[3.] For all $k\in\Gamma$, $c_J(x,\sigma')=c_J(x+k,\sigma)$ id $\sigma'(y)=\sigma(y+k)$ for all $y$.
	\item[4.] Detailed balance: for all $x\in\Gamma$, and all $\sigma$ 
	\begin{align*}
	\exp\left(-\sum_{A\ni x}J_A(\sigma)\right)c_J(x,\sigma)=c_J(x,\sigma^x)\exp\left(   -\sum_{A\ni x}J_A(\sigma^x)\right)\,.
	\end{align*}
\end{itemize} 
These assumptions constitute sufficient conditions for the corresponding Markov process to have the Gibbs states over $\Lambda$ as stationary points. Next, we introduce the notion of a quantum embedding of the aforementioned classical Glauber dynamics. This is the Lindbladian of corresponding Lindblad operators given by 
\begin{align}\label{Lindbladops}
L_{x,\eta}:=\sqrt{c_J(x,\eta)}\,|\eta^x\rangle\langle \eta|\otimes \Id\,,~~~\forall x\in\Lambda,\,\eta\in\Omega_{b_x(r)}\,.
\end{align}
It was shown in \cite{Cubitt2015} that such a dynamics is KMS-symmetric with respect to the state $\mu_\Lambda^\tau$ as embedded into the computational basis. Moreover, the set of fixed points in the Schr\"{o}dinger picture corresponds to the convex hull of the set of Gibbs states over $\Lambda$, $\{\mu_\Lambda^\tau|\tau\in\Omega_{\Lambda^c}\}$. In the Heisenberg picture, this implies that the fixed-point algebras $\cF(\cL_A)$ are expressed as 
\begin{align}\label{fixedpoints}
\cF(\cL_A):=\bigoplus_{\omega\in\Omega_{\partial A}}\,|\omega\rangle\langle\omega|_{\partial A}\otimes \Id_{A}\otimes \cB(\cH_{A_\partial^c})\,.
\end{align}
Equivalently, 
\begin{align}\label{classcond}
E_{A*}(\rho)=\sum_{\omega\in\Omega_{\partial A}}\,|\omega\rangle\langle\omega|_{\partial A}\otimes \sigma^\omega_{A}\otimes \tr_A(\langle \omega|\rho|\omega\rangle)
\equiv\sum_{\omega\in\Omega_{\partial A}}\,|\omega\rangle\langle\omega|_{\partial A}\otimes E_{A*}^\omega(\langle\omega|\rho|\omega\rangle) \,,
\end{align}
where $\sigma^\omega_A$ denotes the Gibbs state $\mu^\omega_A$ embedded into the computational basis.

With this expression at hand we can prove that classical Hamiltonians over quantum systems satisfy the same approximate tensorization than in the classical case.

\begin{theorem}\label{prop_classical_inftyT}
 Let $A,B\subset\Lambda$. Then, at $\beta=0$, $\cN_A$ and $\cN_B$ form a commuting square, that is,
 \begin{equation}\label{eq_prop_classical_inftyT}
  E_A^{\beta=0}\circ E_B^{\beta=0}=E_B^{\beta=0}\circ E_A^{\beta=0} = E_{A\cup B}^{\beta=0}\,
 \end{equation}
  and consequently, for all $\rho\in\cD(\cH_\Lambda)$,
 \begin{equation}\label{eq_prop_classical_inftyT2}
  D\big(\rho\|E_{A\cup B*}^{\beta=0}(\rho)\big)\leq D\big(\rho\|E_{A*}^{\beta=0}(\rho)\big)+D\big(\rho\|E_{B*}^{\beta=0}(\rho)\big)\,.
 \end{equation}
 At finite temperature $\beta>0$, $\operatorname{AT}(c,0)$ holds with
 \begin{equation}\label{eq_prop_classical_inftyT3}
  c=\frac1{1-c_1}\qquad\text{where}\quad c_1=\sup_{\omega\in\Omega_{\partial (A\cup B)}}\,\|E_A^\omega\circ E_B^\omega-E_{A\cup B}^\omega\,:\,\mathbb L_1(\mu_{A\cup B}^\omega)\to\mathbb L_\infty\|\,.
 \end{equation}
\end{theorem}

\begin{proof}
\Cref{eq_prop_classical_inftyT} is a direct consequence of the definition of the conditional expectations at $\beta=0$: i.e. $E_{A}^{\beta=0}=\Id_A\otimes \tr_A$. In order to prove that AT$(c,0)$ holds at positive temperature, we use our main result on approximate tensorization based on Pinching techniques, namely Theorem \ref{theo_AT_pinching}. More specifically, for every $\rho \in \cD (\cH_{\Lambda})$, we denote $\rho_{\mathcal{M}}:= E_{A \cup B^*} (\rho)$ and apply \Cref{eqgeneral} to $\eta=\cP_{\rho_\cM}(\rho)$. Thus, we only need to check that $D_{\max}\big(E_{A*}\circ E_{B*}(\rho)\|E_{A*}\circ E_{B*}(\eta)\big)=0$. We denote by $\cP_A$ the pinching map on the computational basis on a subset $A$ of $\Lambda$. By a simple computation we see that $\cP_{(A\cup B)_\partial}\circ\cP_{\rho_\cM}=\cP_{(A\cup B)_\partial}$ and 
\[E_{A*}\circ E_{B*}= E_{A*}\circ E_{B*}\circ \cP_{(A\cup B)_\partial}\,,\]
so that $E_{A*}\circ E_{B*}(\rho)=E_{A*}\circ E_{B*}(\cP_{\rho_\cM}(\rho))$, which completes the proof.
\end{proof}

 In \Cref{prop_classical_inftyT}, we have shown that strong approximate tensorization AT($1,0$) holds at infinite temperature for classical Hamiltonians.
 However, let us remark that it is not clear (and we strongly believe the opposite) that this remains true for non-classical commuting Gibbs states. A first idea to support this intuition has been shown in \Cref{prop_cond_ABdisjoint}. We leave a thorough study of this fact for future work.
 
\section{Outlook}\label{sec:conclusion}

In this paper, we introduce and study an extension of the celebrated strong subadditivity of the entropy: given algebras $\cN=\cN_1\cap\cN_2$, $\cN_1,\cN_2\subseteq\cM$, with corresponding conditional expectations $E_1:\cM\to \cN_1$, $E_2:\cM\to \cN_2$ and $E_\cN:\cM\to \cN$, there exist constants $c\ge 1$ and $d\ge 0$ such that 
\begin{align}\label{approximatetens}
    D(\rho\|E_*^\cM(\rho))\le c\,\Big(D(\rho\|E_{1*}(\rho))+D(\rho\|E_{2*}(\rho))\Big)+d\qquad \forall\rho\,.
\end{align}
In analogy with its classical analogue, we dubbed this inequality \textit{approximate tensorization} of the relative entropy.

Since the first submission of this paper, \eqref{approximatetens} has found several extensions and applications in the fields of quantum information theory and many body quantum systems: first, the inequality was used to derive the first proof of the positivity of the modified logarithmic Sobolev inequality constant independently of the system size for Gibbs states of nearest neighbour commuting Hamiltonians on a regular lattice \cite{capel2020modified}. For this specific class of Gibbs states, the authors showed that the analysis can indeed be reduced to the case of states $\rho$ for which the additive error term in \Cref{theo_AT_pinching} vanishes, hence providing a direct application to our main result. 

More recently, \cite{gao2021spectral} (as well as a new version of \cite{nick2019scoopingpaper}) proved a strong approximate tensorization result with multiplicative constant depending on the $\mathbb{L}_2$ clustering of the conditional expectations as well as the dimension of the system. Their approximate tensorization was then used to find asymptotically tight exponential entropic decay to equilibrium for various models of noise including quantum Markov semigroups generated by classical graph Laplacians, approximate $k$-designs, or the quantum Kac master equation. In their extension of \eqref{approximatetens}, the noisy system can also be coupled to an arbitrarily large noiseless environment. Although providing a tight approximate tensorization result in the sense that $d=0$ and that it reduces to the exact tensorization in the commuting square setting, their bound however still provides a poor control of the multiplicative constant in the context of Gibbs samplers. We expect that both methods combined will prove useful in proving the uniform positivity of the MLSI constant for generic quantum Gibbs samplers in the near future. Indeed, these techniques, together with a version of Theorem \ref{theo_AT_pinching}, will be used soon to derive positivity of a MLSI for Davies generators in 1D systems \cite{BardetCapelGaoLuciaPerezGarciaRouze-Davies1DMLSI-2021}.

\paragraph{Acknowledgements} IB was supported by Region Ile- de-France in the framework of DIM SIRTEQ. AC was partially supported by a La Caixa-Severo Ochoa grant (ICMAT Severo Ochoa project SEV-2011-0087, MINECO) and acknowledges support from MINECO (grant MTM2017-88385-P), from Comunidad de Madrid (grant QUITEMAD-CM, ref. P2018/TCS-4342) and from ICMAT Severo Ochoa project SEV-2015-0554 (MINECO).  CR is grateful to Federico Pasqualotto for useful discussions, and acknowledges financial support from the TUM university Foundation Fellowship. CR and AC acknowledge funding by the Deutsche Forschungsgemeinschaft (DFG, German Research Foundation) under Germanys Excellence Strategy EXC-2111 390814868. This project has received funding from the European Research Council (ERC) under the European Union’s Horizon 2020 research and innovation programme (grant agreement No 648913).

\bibliography{library}

\appendix

\section{Conditional expectations on fixed-points of Markovian evolution}\label{sec:condexp}

In this section, we consider conditional expectations arising from Petz recovery maps and from Davies generators as introduced in \Cref{sect_twoexamples}. The main result, \Cref{theo_equal_cond}, states that the corresponding conditional expectations coincide. 

\subsection{Conditional expectations generated by a Petz recovery map}

Here, we further discuss the notion of conditional expectations coming from the Petz recovery map. The discussion is largely inspired by some results in \cite{carlen2017recovery}. 

Let $\sigma$ be a faithful density matrix on the finite-dimensional algebra $\cN$ and let $\cM\subset\cN$ be a subalgebra. We denote by $E_\tau$  the conditional expectation onto $\cM$ with respect to the completely mixed state (i.e. $E_\tau$ is self-adjoint with respect to the Hilbert-Schmidt inner product). Let us recall the notations and notions introduced in Section \ref{subsec:conditional_expectations} regarding the adjoint of the Petz recovery map and the conditional expectation constructed from it. We show below the form that these concepts take for a bipartite system.

\begin{example}\label{ex_bipartite}
 Our main example is the case of a bipartite system $AB$. In this case, $\cN=\cB(\cH_{AB})$ and $\cM=\Id_{\cH_A}\otimes\cB(\cH_{B})$. Let $\sigma=\sigma_{AB}$ be a faithful density matrix on $AB$. The partial trace with respect to $\cH_A$ is an example of a conditional expectation $E_\tau$ which is not compatible with $\sigma_{AB}$, in general. With this choice, we obtain:
 \begin{align*}
  & \sigma_\cM=\sigma_B\,,\\
  & \mathcal A_{\sigma_{AB}}(X)=\sigma_{B}^{-\frac12}\,\tr_A[\sigma_{AB}^{\frac12}\,X\,\sigma_{AB}^{\frac12}]\,\sigma_{B}^{-\frac12}\,,\qquad\forall X\in\cB(\cH_{AB})\, , \\
  & \cR_{\sigma_{AB}}(\rho_B)=\sigma_{AB}^{\frac12}\sigma_{B}^{-\frac12}\,\rho_B\,\sigma_{B}^{-\frac12}\sigma_{AB}^{\frac12}\,,\qquad\forall \rho_{AB}\in\cD(\cH_{AB})\,,
 \end{align*}
 where here we identify an operator $X_B$ with $\Id_A\otimes X_B$ for sake of simplicity. An important remark is that, in general, $E_{\sigma_{AB}*}$ is not a recovery map.
\end{example}

We are now ready to state a first technical proposition, whose content is mostly contained in \cite{carlen2017recovery}.

\begin{proposition}\label{prop_condexp}
 Let $\rho$ be a density matrix on $\cN$. Then the following assertions are equivalent:
\begin{enumerate}
 \item $D(\rho\|\sigma)=D(\rho_\cM\|\sigma_\cM)$;
 \item $\rho=\cR_\sigma(\rho_\cM)$;
 \item $\rho=E_{\sigma*}(\rho)$;
  \item $D(\rho\|E_{\sigma*}(\rho))=0$.
 \item $D(\rho\|\sigma)=D(E_{\sigma*}(\rho)\|\sigma)$;
\end{enumerate}
 \end{proposition}

Remark that $(1)\Leftrightarrow(2)$ is Petz condition for equality in the data processing inequality. The equivalence $(3)\Leftrightarrow(4)$ is obvious, and $(4)\Leftrightarrow(5)$ is a consequence of Lemma 3.4 in \cite{junge2019stability}:
\begin{equation}\label{very_nice_lemma}
D(\rho\|\sigma)-D(E_{\sigma*}(\rho)\|\sigma)=D(\rho\|E_{\sigma*}(\rho))\,. 
\end{equation}
We shall now give a direct proof of $(2)\Leftrightarrow(3)$. 

\begin{proof}[Proof of \Cref{prop_condexp}]
We only prove $(2)\Leftrightarrow(3)$. Note that for $X\in\cN$, by definition $X=\mathcal A_\sigma(X)$ iff $X=E_\sigma(X)$. Then let $\rho$ be a density matrix on $\cN$ and define $X=\sigma^{-\frac12}\,\rho\,\sigma^{-\frac12}$. We have:
\begin{align*}
 \rho=\cR_\sigma(\rho_\cM)
 & \Leftrightarrow X=\mathcal A_\sigma(X) \\
 & \Leftrightarrow X=E_\sigma(X) \\
 & \Leftrightarrow \rho= \sigma^{\frac12}\,X\,\sigma^{\frac12} = \sigma^{\frac12}\,E_\sigma(X)\,\sigma^{\frac12}=E_{\sigma*}(\sigma^{\frac12}\,X\,\sigma^{\frac12})=E_{\sigma*}(\rho)\,.
\end{align*}
where in the last line we use property 3 in Proposition \ref{prop_condexp}.
\end{proof}

It would be interesting to compare the two notions of ``conditional'' relative entropies $D(\rho\|\sigma)-D(\rho_\cM\|\sigma_\cM)$ (introduced in \cite{[CLP18],[CLP18a],BardetCapelLuciaPerezGarciaRouze-HeatBath1DMLSI-2019}) and $D(\rho\|E_{\sigma*}(\rho))$. This is the content of the following proposition.

\begin{proposition}\label{prop_compar_cond}
 For any state $\eta\in\cD(\cN)$ such that $E_{\sigma*}(\eta)=\eta$ and any state $\rho\in\cD(\cN)$, we have
 \begin{equation}\label{eq_prop_compar_cond1}
  D(\rho\|\sigma)-D(\rho_\cM\|\sigma_\cM)=D(\rho\|\eta)-D(\rho_\cM\|\eta_\cM)\,,
 \end{equation}
i.e. the difference of relative entropies does not depend on the choice of the invariant state for $E_\sigma$. Consequently,
\begin{equation}\label{eq_prop_compar_cond2}
 D(\rho\|\sigma)-D(\rho_\cM\|\sigma_\cM)\leq D(\rho\|E_{\sigma*}(\rho))\,.
\end{equation}
\end{proposition}

\begin{proof}
 \Cref{eq_prop_compar_cond2} is a direct consequence of \Cref{eq_prop_compar_cond1} when applied to $\eta=E_{\sigma*}(\rho)$, so we focus on the first equation (remark that it can be seen as a counterpart of \Cref{very_nice_lemma} for the difference of relative entropies). To this end, we need the following state $\sigma_{\tr}$ defined in \cite{BarEID17} and heavily exploited in \cite{bardet2018hypercontractivity}:
 \[\sigma_{\tr}=E_{\sigma*}\Big(\frac{\mathds 1}{d_\cH}\Big)\,.\]
 It has the property that for all $X\in\cF(\mathcal A_\sigma)$, $[X,\sigma_{\tr}]=0$ (see Lemma 3.1 in \cite{BarEID17}). Then it is enough to prove that for all $\eta\in\cD(\cN)$ such that $E_{\sigma*}(\eta)=\eta$, we have:
 \[D(\rho\|\sigma_{\tr})-D(\rho_\cM\|(\sigma_{\tr})_\cM)=D(\rho\|\eta)-D(\rho_\cM\|\eta_\cM)\,.\]
 Now any such $\eta$ can be written $\eta=X\sigma_{\tr}$ with $X\in\cF(\mathcal A_\sigma)$. Remark that by definition of $\cF(\mathcal A_\sigma)$, $X\in\cM$  so that $E_\tau(\eta)=X E_\tau(\sigma_{\tr})$. Using the commutation between $X$ and $\sigma_{\tr}$ and developping the RHS of the previous equation we get the result.
\end{proof}

\subsection{Davies semigroups}\label{subsec:clustering-of-correlations}

Here we consider the conditional expectation associated to the Davies dynamics that was presented in Section \ref{sect_twoexamples}. Our first result is a characterization of the fixed-point algebra in the Davies case.
\begin{proposition}\label{prop_fixedpoint_Davies}
 One has
 \begin{equation}\label{eq_prop_fixedpoint_Davies}
  \cF(\cL^{\operatorname{D},\beta})=\{\sigma^{it}\,S_\alpha\,\sigma^{-it}\,;\,\forall\,t\geq0,\,\forall \alpha\}'\,,
 \end{equation}
 where the notation $\{ \cdot \}'$ denotes the centralizer of the set. 
\end{proposition}
\begin{proof}
	We recall that $\cF(\cL^{\operatorname{D,\beta}})=\{S_\alpha(\omega)\}'$. Hence, since $\sigma^{it}S_\alpha \sigma^{-it}$ can be expressed as a linear combination of the $S_\alpha(\omega)$'s by \Cref{eq!}, it directly follows that
	\begin{align*}
	\cF(\cL^{\operatorname{D,\beta}})\subseteq \{\sigma^{it}\,S_\alpha\,\sigma^{-it}\,;\,t\geq0\}'
	\end{align*}
	To prove the opposite direction, we let $X\in  \{\sigma^{it}\,S_\alpha\,\sigma^{-it}\,;\,t\geq0\}'$. This means in particular that, for all $t\in\RR$, and all $\alpha$:
	\begin{align}\label{eq:sumomega}
	[X, \sigma^{it}S_{\alpha}\sigma^{-it}]=\sum_{\omega}\,\e^{it\omega}\,[X,S_{\alpha}(\omega)]=0\,.
	\end{align}
	Since the equation holds for all $t\in\RR$, we can differentiate it $N\equiv|\{\omega\}|$ times at $0$  to get that, for any $0\le n\le N-1$:
	\begin{align*}
	\sum_{\omega}\,\omega^n\,[X,S_\alpha(\omega)]=0\,.
	\end{align*}
Using an arbitrary labelling of the $N$ distinct frequencies $\omega_1,...,\omega_N$, the resulting $N$ linear equations can be rewritten as
\begin{align*}
\begin{pmatrix}
1 & 1& 1 & \dots & 1 \\ 
\omega_1& \omega_2 &\omega_3 & \dots & \omega_N \\
\omega_1^2&\omega_2^2 &\omega_3^2 & \dots & \omega_N^2 \\
\hdotsfor{5} \\
\omega_1^{N-1} & \omega_2^{N-1} &\omega_3^{N-1} & \dots &\omega_N^{N-1}
\end{pmatrix}\,\begin{pmatrix}
[X,S_\alpha(\omega_1)]\\
[X,S_\alpha(\omega_2)]\\
[X,S_\alpha(\omega_3)]\\
\hdotsfor{1}\\
[X,S_\alpha(\omega_N)]
\end{pmatrix}=0
\end{align*}
Since all the frequencies $\omega_i$ are distinct, their Vandermonde matrix is invertible. Hence, $[X,S_\alpha(\omega)]=0$ for all $\omega$, so that $X\in\cF(\cL^{\operatorname{D},\beta})$.
\end{proof}
Combining this result with a result from \cite{carlen2017recovery}, we can finally show the result stated in \Cref{theo_equal_cond} that the conditional expectations in the Davies and the Petz cases coincide.

\begin{proof}[Proof of \Cref{theo_equal_cond}]
 First, we remark that both conditional expectations are self-adjoint with respect to the $\sigma$-KMS inner product. Therefore, by uniqueness of the conditional expectation, it is enough to prove that $\cF(\cL^{\operatorname{D},\beta})=\cF(A_{\sigma})$. The analysis of the algebra $\cF(A_{\sigma})$ was carried out in \cite{carlen2017recovery}\footnote{Compared to \cite{carlen2017recovery}, the role of $\rho$ and $\sigma$ is exchanged. The result nevertheless stays the same, as can be readily checked from their proof.}. In particular, they proved (Theorem 3.3) that $\cF(A_{\sigma})$ is the largest $*$-sub-algebra of $\cM$ left-invariant by the modular operator. From this characterization, it is easy to see that $\cF(\cL^{\operatorname{D},\beta})\subseteq \cF(\cA_\sigma)$: indeed $\cF(\cL^{\operatorname{D},\beta})=\{S_\alpha(\omega)\}'\subseteq \{S_\alpha\}'\equiv \cM$. Moreover, for any $X\in \cF(\cL^{\operatorname{D},\beta})$
 \begin{align*}
 [\Delta_\sigma(X),S_\alpha(\omega)]&=\sigma\,X\,\sigma^{-1}\,S_\alpha(\omega)-S_\alpha(\omega)\sigma X\sigma^{-1}\\
 &=\e^{-\beta\omega }\,\big( \sigma XS_{\alpha}(\omega)\sigma^{-1}-\sigma\,S_\alpha(\omega)X\sigma^{-1}  \big)\\
 &=\e^{-\beta\omega }\,\sigma\,[X,S_\alpha(\omega)]\,\sigma^{-1}\\
 &=0\,.
 \end{align*} 
 It remains to show that any $*$-sub-algebra $\mathcal{V}$ of $\cM$ which is invaraint by $\Delta_\sigma$ is contained in $\cF(\cL^{\operatorname{D},\beta})$. This directly follows from (\ref{eq_prop_fixedpoint_Davies}): since for any $X\in\mathcal{V}$, $\Delta(X)\in\mathcal{V}$, we have that
 \begin{align*}
 [X,\sigma^{it}\,S_\alpha\,\sigma^{-it}]&=X\,\sigma^{it}\,S_\alpha\,\sigma^{-it}-\sigma^{it}\,S_\alpha\,\sigma^{-it}X\\
 &=\sigma^{it}\Delta_{\sigma}^{-it}(X)S_\alpha\sigma^{-it}-\sigma^{it}S_\alpha \Delta_{\sigma}^{-it}(X)\sigma^{-it}\\
 &=\sigma^{it}\,[\Delta_{\sigma}^{-it}(X),S_\alpha]\,\sigma^{-it}\\
 &=0\,
 \end{align*}
 and the result follows.
\end{proof}

\section{Proofs}\label{appendix:proofs}

\subsection{Proof of \Cref{HSAT}}\label{subsec:proof_change_measure}

\begin{proof}[Proof of \Cref{HSAT}]
	The proof of this result relies on the Holley-Stroock perturbative argument for the Lindblad relative entropy proved in \cite{junge2019stability}. This entropic distance is defined for two positive semi-definite operators $X,Y\in\cB(\cH)$ such that $Y$ is full rank as 
	\[D_{\operatorname{Lin}}(X\|Y):=\tr[X(\log X-\log Y)]-\tr[X]+\tr[Y]\,.\]
	Next, we use a direct adaptation of the proof of Proposition 4.2 in \cite{junge2019stability} in order to relate the Lindblad relative entropies $D_{\operatorname{Lin}}(\rho\|E_{*}^\cM(\rho))$ and $D_{\operatorname{Lin}}(\rho\|E_{*}^{(0),\cM}(\rho))$. More precisely, we have that for any positive, semidefinite operators $X$ and $Y$, 
	\begin{align}\label{eq_proof_changeofvariable}
	\frac{1}{\lambda_{\max}(\sigma)\,d_\cH} \,D_{\operatorname{Lin}}(\Gamma_{\sigma}(X)\| E_{*}^{\cM}(\Gamma_\sigma(Y)) )\le  D_{\operatorname{Lin}}(X\|E_{*}^{(0),\cM}(Y))\le \frac{1}{\lambda_{\min}(\sigma)d_{\cH}}\,D_{\operatorname{Lin}}(\Gamma_{\sigma}(X)\| E_{*}^{\cM}(\Gamma_\sigma(Y)) )\,,
	\end{align}
	where $\lambda_{\min}(\sigma)$, resp. $\lambda_{\max}(\sigma)$, denotes the smallest, resp. largest, eigenvalue of the state $\sigma$. In words, the proof of \cite[Proposition 4.2]{junge2019stability} consists in the observation that the conditional expectations $E^{(0),\cM}$ and $E^{\cM}$ are related via $d_\cH E^{\cM}_*=\Gamma_\sigma^{-1}E^{(0),\cM}_*$, together with the monotonicity of $D_{\operatorname{Lin}}$ under completely positive, trace non-increasing maps. Analogous inequalities hold for $E_1$ and $E_2$. 
	Similarly to what is done for classical spin systems in \cite{ledoux2001logSobolev}, the previous inequality can be rewritten in the following way. Consider the generalization of the relative entropy for $X=\Gamma_\sigma^{-1}(\rho)$ given by:
\[\operatorname{Ent}_{1,\cM}(X):=D(\rho\|E_*^\cM(\rho))\,,\]
with analogous expressions for $\cN_1$ and $\cN_2$ with their respective conditional expectations. Then, we can express this relative entropy as an infimum over $D_{\operatorname{Lin}}$. Indeed, Lemma 3.4 in \cite{junge2019stability} states that for all full-rank positive semi-definite $Y\in\cM$,
 \begin{equation}\label{eq_proof_HS}
  D_{\operatorname{Lin}}(\rho\|\Gamma_{\sigma}(Y))=D_{\operatorname{Lin}}(\rho\|E_*^\cM(\rho))+D_{\operatorname{Lin}}(E_*^\cM(\rho)\|\Gamma_{\sigma}(Y))\,.
 \end{equation}
It shows in particular that $D_{\operatorname{Lin}}(\rho\|\Gamma_{\sigma}(Y))\geq D_{\operatorname{Lin}}(\rho\|E_*^\cM(\rho))=\operatorname{Ent}_{1,\cM}(X)$, with equality for $Y=E^\cM(X)$. Thus, we obtain
\begin{equation}\label{eq_var_entropy}
 \operatorname{Ent}_{1,\cM}(X)=\underset{Y\in \cM\,,Y>0}{\inf}\,D_{\operatorname{Lin}}(\rho\|\Gamma_{\sigma}(Y))\,
\end{equation}
and optimizing over all $Y$ we can rewrite \Cref{eq_proof_changeofvariable} as
\begin{align*}
	\frac{1}{\lambda_{\max}(\sigma)\,d_\cH} \,\operatorname{Ent}_{1,\cM}(X)\le  D(\rho\|E_{*}^{(0),\cM}(\rho))\le \frac{1}{\lambda_{\min}(\sigma)d_{\cH}}\,\operatorname{Ent}_{1,\cM}(X)\,.
	\end{align*}
Finally, using the approximate tensorization at infinite temperature and rearranging the terms leads to the result:
\begin{equation*}
 \operatorname{Ent}_{1,\cM}(X) \leq \frac{\lambda_{\max}(\sigma)}{\lambda_{\min}(\sigma)}\,\big( \operatorname{Ent}_{1,\cN_1}(X)+\operatorname{Ent}_{1,\cN_2}(X) \big) +\lambda_{\max}(\sigma)\,d_\cH\,d\,.
\end{equation*}
\end{proof}

\subsection{Proof of \Cref{propboundsd1d2}}\label{subsec:proof_estimation_c,d}

\begin{proof}[Proof of Proposition \ref{propboundsd1d2}]
	We first proceed by proving the bound $d\leq d_1+d_2$. For all $\rho\in\cD(\cN)$, we can use the chain rule on the max-relative entropy to obtain:
	\begin{align*}
	& D_{\max}\big(E_{1*}\circ E_{2*}(\rho)\|E_{1*}\circ E_{2*}(\eta)\big) \\
	& ~~~~~~~~~~~~~~~~~~~~~~~~~ \leq D_{\max}\big(E_{1*}\circ E_{2*}(\rho)\|E_{1*}\circ E_{2*}(\cP_\cM(\rho))\big) +D_{\max}\big(E_{1*}\circ E_{2*}(\cP_\cM(\rho))\|E_{1*}\circ E_{2*}(\eta)\big)\\
	& ~~~~~~~~~~~~~~~~~~~~~~~~~  \leq d_1 + D_{\max}\big(\cP_\cM(\rho)\|\eta\big)\,,
	\end{align*}
    where the second inequality follows from the data processing inequality for $D_{\max}$. Then 
    \[D_{\max}\big(\cP_\cM(\rho)\|\eta\big)\leq\max_{i\in I_\cM}D_{\max}\big(\cP_\cM(\rho)^{(i)}\|\eta^{(i)}\big)\,,\]
	where we write $A^{(i)}:=P_i\,A\,P_i$ for any $A\in\cB(\cH)$. This last $D_{\max}$ is exactly $I_{\max}\big( \cH_i:\cK_i  \big)_{\rho^{(i)}}$ after minimizing on $\eta$.	We are left with proving the two separate bounds on $d_1$ and $d_2$ respectively. The first bound is a simple consequence of the data processing inequality for $D_{\max}$ and the Pinching inequality. The second bound is a consequence of Lemma B.7 in \cite{berta2011quantum}.
\end{proof}

\subsection{Proofs of \Cref{theo_L2clust} and \Cref{theo_transference}}\label{subsec:proof_clustering}

Before proving \Cref{theo_L2clust}, we need to prove a technical lemma.
\begin{lemma}\label{lemma:L2clustering-invariant}
	Given a conditional expectation $E:\cN\to\cM\subset \cN\subset \cB(\cH)$ that is invariant with respect to two different full-rank states, $\rho$ and $\sigma$, the following holds:
	\begin{align*}
	\Gamma_\rho^{1/2}\circ E\circ \Gamma_{\rho}^{-1/2}=\Gamma_\sigma^{1/2}\circ E\circ \Gamma_\sigma^{-1/2}
	\end{align*}
\end{lemma}

\begin{proof}[Proof of \Cref{lemma:L2clustering-invariant}]
	Since we are in finite dimension, the von Neumann algebra $\cM$ takes the following form:
	\begin{align*}
	\cM=\bigoplus_i\,\cB(\cH_i)\otimes\Id_{\cK_i}\,,
	\end{align*}
	for some decomposition $\cH:=\bigoplus_i\,\cH_i\otimes\cK_i$ of $\cH$. Therefore, since $\rho$ and $\sigma$ are invariant stats of $E$, they can be decomposed as follows:
	\begin{align*}
	\rho=\bigoplus_{i}\, \rho_i\otimes\tau_i\,,~~~	\sigma=\bigoplus_{i}\, \sigma_i\otimes\tau_i\,,
	\end{align*}
	for given positive definite operators $\sigma_i$, $\rho_i$ and where $\tau_i$ is given by $\Id_{\mathcal{K}_i}/d_{\mathcal{K}_i}$. Hence, 
	\begin{align*}
	\rho^{-1/4}\sigma^{1/4}=\bigoplus_i\,\rho_i^{-1/4}\sigma_i^{1/4}\otimes\Id_{\cK_i}\in\cN.
	\end{align*}	
	Then, it is clear that the following string of identities hold for all $Y\in\cB(\cH)$:
	\begin{align*}
	\rho^{-1/4}\,\sigma^{1/4}\,E\big[\sigma^{-1/4}\rho^{1/4}\,Y\,\rho^{1/4}\sigma^{-1/4}\big]\,\sigma^{1/4}\rho^{-1/4}&=E\big[\rho^{-1/4}\,\sigma^{1/4}\sigma^{-1/4}\rho^{1/4}\,Y\,\rho^{1/4}\sigma^{-1/4}\sigma^{1/4}\rho^{-1/4}\big]\\
	&=E[Y]\,.
	\end{align*}
	The result follows after choosing $Y=\rho^{-1/4}X\rho^{-1/4}$.
	
\end{proof}

Now we can proceed to the proof of \Cref{theo_L2clust}.

\begin{proof}[Proof of \Cref{theo_L2clust}]
We begin with proving that the property of strong $\mathbb{L}_2$ clustering of correlations is independent of the invariant state, thanks to \Cref{lemma:L2clustering-invariant}. Indeed, if we choose $Y:= \Gamma_\sigma^{-1/2}(X)$ and call $X':= \Gamma_{\sigma'}^{1/2}(Y)$, it is clear that 
	\begin{equation*}
	\norm{X}_{\mathbb{L}_2(\sigma)}^2= \norm{Y}_2^2 \; \; \text{ and } \; \;\norm{Y}_2^2 = \norm{X'}_{\mathbb{L}_2(\sigma')}^2.
	\end{equation*}
	Therefore, we have the following chain of identities:
	\begin{align*}
	\sup_{X\in\cN} \, \frac{\operatorname{Cov}_{\cM,\,\sigma}(E_1[X],E_2[X])}{\|X\|_{\mathbb{L}_2(\sigma)}^2} & = \sup_{X\in\cN} \, \frac{\langle X,\,E_1\circ E_2[X]-E^\cM[X] \rangle_{\sigma} }{\|X\|_{\mathbb{L}_2(\sigma)}^2}\\
	&=\sup_{Y\in\cN} \, \frac{\langle \Gamma_{\sigma}^{-1/2}(X),\,E_1\circ E_2[\Gamma_{\sigma}^{-1/2}(X)]-E^\cM[\Gamma_{\sigma}^{-1/2}(X)] \rangle_{\sigma} }{\|Y\|_{2}^2}\\
	&=\sup_{Y\in\cN}\, \frac{\langle X,\,\Gamma_{\sigma}^{1/2}(E_1\circ E_2-E^\cM)[\Gamma_{\sigma}^{-1/2}(X)] \rangle_{\operatorname{HS}} }{\|Y\|_{2}^2}\\
	&=\sup_{Y\in\cN} \, \frac{\langle \Gamma_{\sigma'}^{-1/2}(X),\,E_1\circ E_2[\Gamma_{\sigma'}^{-1/2}(X)]-E^\cM[\Gamma_{\sigma'}^{-1/2}(X)] \rangle_{\sigma'} }{\|Y\|_{2}^2}\\
	&=\sup_{X'\in\cN} \, \frac{\operatorname{Cov}_{\cM,\,\sigma}(E_1[X'],E_2[X'])}{\|X'\|_{\mathbb{L}_2(\sigma)}^2}\,,
	\end{align*}
	where we have used Lemma \ref{lemma:L2clustering-invariant} in the fourth line. 
\end{proof}

In a different direction, we can also provide the proof of \Cref{theo_transference}.

\begin{proof}[Proof of \Cref{theo_transference}]
Point 1. is straigthforward so we focus on point 2. As already mentioned, strong $ \mathbb L_2$ clustering implies $\operatorname{cond\mathbb{L}_2}(c_2)$, so we only need to prove the other implication. Now assume that $\operatorname{cond\mathbb{L}_2}(c_2)$ holds with a constant $c_2$ and take $X\in\cD(\cN)$. We write $T=E_1\circ E_2-E^\cM$. Remark that, according to the decomposition of $\cM$ given in \Cref{eq_decomp} and exploiting \Cref{eq_d1=0}, $T$ acts on $X$ as:
 \begin{equation}\label{eq_decomp_T}
   T(X)=\sum_{i\in I_\cM}\,\left(\Id_{\cB(\cH_i)}\otimes T^{(i)}\right)(P_i\,X\,P_i)\,,
  \end{equation}
 where $T^{(i)}$ acts on $\cB(\cK_i)$ and where the $P_i$ are the orthogonal projections on $\cH_i\otimes\cK_i$.
   
 Consider now the Hilbert-Schmidt decomposition of $P_iXP_i$ with respect to $(\cB(\cH_i),\langle\cdot , \cdot\rangle_{\sigma_i})$ and $(\cB(\cK_i),\langle\cdot , \cdot\rangle_{\tau_i})$:
 \[P_i\,X\,P_i=\sum_\alpha\,f_\alpha^{(i)}\otimes g_\alpha^{(i)}\,.\]
 Thus we have
  \[\| P_i\,X\,P_i\|_{\mathbb L_2(\sigma_i\otimes\tau_i)}^2=\sum_\alpha\,\|f_\alpha^{(i)}\|_{\mathbb L_2(\sigma_i)}^2\,\|g_\alpha^{(i)}\|_{\mathbb L_2(\tau_i)}^2\,,\]
 and therefore
    \begin{align*}
    \|T(X)\|_{\mathbb L_2(\sigma_i\otimes\tau_i)}^2
    & = \sum_{i\in I_\cM}\,\left\|\left(\Id_{\cB(\cH_i)}\otimes T^{(i)}\right)(P_i\,X\,P_i)\right\|_{\mathbb L_2(\sigma_i\otimes\tau_i)}^2 \\
    & = \sum_{i\in I_\cM}\,\left\|\sum_\alpha f_\alpha^{(i)}\otimes T^{(i)}(g_\alpha^{(i)})\right\|_{\mathbb L_2(\sigma_i\otimes\tau_i)}^2 \\
    & = \sum_{i\in I_\cM}\,\sum_\alpha \|f_\alpha^{(i)}\|_{\mathbb L_2(\sigma_i)}^2\,\|T^{(i)}(g_\alpha^{(i)})\|_{\mathbb L_2(\tau_i)}^2 \\
    & \leq c_2 \left\|\sum_{i\in I_\cM}\,P_i\,X\,P_i\right\|_{\mathbb L_2(\sigma)}^2 \\
    & \leq c_2 \|X\|_{\mathbb L_2(\sigma)}^2\,,
    \end{align*}
 where in the third line we use that $(f_\alpha^{(i)})_\alpha$ is an orthogonal family for every $i\in I_\cM$. This shows that
    \[\|E_1\circ E_2-E^\cM\,:\,L_2(\sigma)\to L_2(\sigma)\|\leq c_2\,,\]
 which is equivalent to strong $L_2$ clustering.
\end{proof}

\end{document}